\documentclass[11pt]{article}
\usepackage{verbatim}
\usepackage{fullpage}
\usepackage{odonnell}

\newcommand{\onote}[1]{\footnote{{\bf \color{blue}Ryan}: {#1}}}
\newcommand{\rnote}[1]{\footnote{{\bf \color{red}Rocco}: {#1}}}

\newcommand{\channel}{\mathcal{C}}
\newcommand{\gap}{\epsilon}

\newcommand{\GM}{\mathrm{GM}}

\usepackage{caption}

\usepackage{todonotes}

\makeatletter
\newtheorem*{rep@theorem}{\rep@title}
\newcommand{\newreptheorem}[2]{
\newenvironment{rep#1}[1]{
 \def\rep@title{#2 \ref{##1}}
 \begin{rep@theorem}\itshape}
 {\end{rep@theorem}}}
\makeatother
\theoremstyle{plain}



\def\colorful{1}

\ifnum\colorful=1
\newcommand{\violet}[1]{{\color{violet}{#1}}}

\newcommand{\blue}[1]{{{\color{blue}#1}}}
\newcommand{\red}[1]{{ {#1}}}

\newcommand{\gray}[1]{{\color{gray}{#1}}}

\fi
\ifnum\colorful=0
\newcommand{\violet}[1]{{{#1}}}

\newcommand{\blue}[1]{{{#1}}}
\newcommand{\red}[1]{{{#1}}}

\newcommand{\gray}[1]{{{#1}}}

\fi

\usepackage{boxedminipage}

\newreptheorem{theorem}{Theorem}
\newtheorem*{theorem*}{Theorem}
\newreptheorem{lemma}{Lemma}
\newreptheorem{proposition}{Proposition}
\newtheorem*{noclaim*}{Claim}

\newcommand{\Del}{\mathrm{Del}}
\newcommand{\Avg}{\mathop{\mathrm{avg}}}
\newcommand{\Bin}{\mathrm{Bin}}

\newcommand{\pparagraph}[1]{{\medskip \noindent {\bf #1}}}

\newcommand{\chill}{\mathrm{frac}}
\newcommand{\ideal}{\mathrm{ideal}}
\newcommand{\Lit}{\mathrm{Littlewood}}
\newcommand{\bounded}{\mathrm{bounded}}

\begin{document}

\title{
Optimal mean-based algorithms for trace reconstruction
}
\author{
Anindya De\thanks{Supported by start-up grant from Northwestern University.}\\
Northwestern University\\
{\tt anindya@eecs.northwestern.edu}
\and
Ryan O'Donnell\thanks{Supported by NSF grant CCF-1618679.}\\
Carnegie Mellon University\\
{\tt odonnell@cs.cmu.edu}
\and
\and Rocco A.~Servedio\thanks{Supported by NSF grants CCF-1420349 and CCF-1563155.}\\
Columbia University \\
{\tt rocco@cs.columbia.edu}
}

\begin{titlepage}

\maketitle

  \begin{abstract}

In the \emph{(deletion-channel) trace reconstruction problem}, there is an unknown $n$-bit \emph{source string} $x$. An algorithm is given access to independent \emph{traces} of~$x$, where a trace is formed by deleting each bit of~$x$ independently with probability~$\delta$. 
The goal of the algorithm is to recover~$x$ exactly (with high probability), while minimizing samples (number of traces) and running time.

Previously, the best known algorithm for the trace reconstruction problem was due to Holenstein~et~al.~\cite{HMPW08}; it uses $\exp(\wt{O}(n^{1/2}))$ samples and running time for any fixed $0 < \delta < 1$.  It is also what we call a ``mean-based algorithm'', meaning that it only uses the empirical means of the individual bits of the traces.  Holenstein~et~al.~also gave a lower bound, showing that any mean-based algorithm must use at least $n^{\wt{\Omega}(\log n)}$ samples.

In this paper we improve both of these results, obtaining matching upper and lower bounds for mean-based trace reconstruction.  For any constant deletion rate $0 < \delta < 1$, we give a mean-based  algorithm that uses  $\exp(O(n^{1/3}))$ time and traces; we also prove that any mean-based  algorithm must use at least $\exp(\Omega(n^{1/3}))$ traces.
In fact, we obtain matching upper and lower bounds even for $\delta$ subconstant and $\rho \coloneqq 1-\delta$ subconstant: when $(\log^3 n)/n \ll \delta \leq 1/2$ the  bound is $\exp(-\Theta(\delta n)^{1/3})$, and when $1/\sqrt{n} \ll \rho \leq 1/2$ the bound is $\exp(-\Theta(n/\rho)^{1/3})$.

Our proofs involve estimates for the maxima of Littlewood polynomials on complex disks.  We show that these techniques can also be used to perform trace reconstruction with random insertions and bit-flips in addition to deletions. We also find a surprising result: for deletion probabilities $\delta > 1/2$, the presence of insertions can actually \emph{help} with trace reconstruction.
\end{abstract}

\thispagestyle{empty}

\end{titlepage}

\section{Introduction}

Consider a setting in which a string $x$ of length $n$ over an alphabet $\Sigma$ is passed through a \emph{deletion channel} that independently deletes each coordinate of $x$ with probability $\delta$.  The resulting string, of length somewhere between $0$ and $n$, is referred to as a \emph{trace} of $x$, or as a \emph{received string}; the original string $x$ is referred to as the \emph{source string}.  The \emph{trace reconstruction problem} is the task of reconstructing $x$ (with high probability) given access to independent traces of $x$.  This is a natural and well-studied problem, dating back to the early 2000's~\cite{Lev01a,Lev01b,BKKM04}, with some combinatorial variants dating even to the early 1970's~\cite{Kalashnik73}. 
However, perhaps surprisingly, much remains to be discovered both about the information-theoretic and algorithmic complexity of this problem.  Indeed, in a 2009 survey~\cite[Section~7]{Mit09}, Mitzenmacher wrote that ``the study of [trace reconstruction] is still in
its infancy''.


Before discussing previous work, we briefly explain why one can assume a binary alphabet without loss of generality. In case of a general $\Sigma$, drawing $O({\frac {\log n }{1-\delta}})$ traces will  with high probability reveal the entire alphabet $\Sigma' \subseteq \Sigma$ of symbols that are present in $x$.   For each symbol $\sigma \in \Sigma'$ we may consider the binary string $x|_\sigma$ whose $i$-th character is $1$ iff $x_i=\sigma$; a trace of $x$ is easily converted into a trace of $x|_\sigma$, so the trace reconstruction problem for $x$ can be solved by solving the binary trace reconstruction problem for each $x|_\sigma$ and combining the results in the obvious way.  
For this reason, our work (and most previous work) focuses on the case of a binary alphabet.

\subsection{Prior work}

As described in~\cite{Mit09}, the trace reconstruction problem can arise in several natural domains, including sensor networks and  biology.  However, the apparent difficulty of the problem means that there is not too much published work, at least on the problem of ``worst-case'' trace reconstruction problem (``worst-case'' in the sense that the source string may be any element of $\zo^n$).  Because of this, several prior authors have considered an ``average-case'' version of the problem in which the source string is assumed to be uniformly random over $\zo^n$ and the algorithm is required to succeed with high probability over the random draw of the traces and over the uniform random choice of $x$.  This average-case problem seems to have first been studied by Batu et al.~\cite{BKKM04}, who showed that a simple efficient algorithm which they call Bitwise Majority Alignment succeeds with high probability for sufficiently small deletion rates $\delta = O(1/\log n)$ using only $O(\log n)$ traces.  Subsequent work of Kannan and McGregor~\cite{KM05} gave an algorithm for random $x$ that can handle both deletions and insertions (both at rates $O(1/\log^2 n)$ as well as bit-flips (with constant probability bounded away from $1/2$) using $O(\log n)$ traces.  Viswanathan and Swaminathan~\cite{VS08} sharpened this result by improving the deletion and insertion rates that can be handled to $O(1/\log n)$.  Finally, \cite{HMPW08} gave a $\poly(n)$-time, $\poly(n)$-trace algorithm for random $x$ that succeeds with high probability for any deletion rate $\delta$ that is at most some sufficiently small absolute constant. 

Several researchers have considered, from an information-theoretic rather than algorithmic perspective, various reconstruction problems that are closely related to the (worst-case) trace reconstruction problem. Kalashnik~\cite{Kalashnik73}  showed that any $n$-bit string is uniquely specified by its \emph{$k$-deck}, which is the multiset of all its length-$k$ subsequences, when $k=\lfloor n/2 \rfloor$; this result was later reproved by Manvel et al.~\cite{MMSSS91}.  Scott~\cite{Scott97} subsequently showed that $k=(1+o(1))\sqrt{n \log n}$ suffices for reconstruction from the $k$-deck for any $x$, and simultaneously and independently Krasnikov and Roditty~\cite{KR97}  showed that $k = \lfloor {\frac {16} 7} \sqrt{n} \rfloor + 5$ suffices.  (McGregor et al.\ observed in~\cite{MPV14} that the result of~\cite{Scott97} yields an information-theoretic algorithm using $\exp(\tilde{O}(n^{1/2}))$ traces for any deletion rate $\delta \leq 1 - O(\sqrt{\log(n)/n})$, but did not discuss the running time of such an algorithm.)  On the other side, successively larger $\Omega(\log n)$ lower bounds on the value of $k$ that suffices for reconstruction of an arbitrary $x \in \zo^n$ from its $k$-deck were given by Manvel~et~al.~\cite{MMSSS91} and Choffrut and Karhum\"{a}ki~\cite{CK97}, culminating in a lower bound of $2^{\Omega(\sqrt{\log n})}$ due to Dud\'{i}k and Schulman~\cite{DS03}.

Surprisingly few algorithms have been given for the worst-case trace reconstruction problem as defined in the first paragraph of this paper.   Batu et al.~\cite{BKKM04} showed that a variation of their Bitwise Majority Alignment algorithm succeeds efficiently using $O(n \log n)$ traces if the deletion rate~$\delta$ is quite low, at most $O(1/n^{1/2+\vareps}).$  Holenstein et al.~\cite{HMPW08} gave a ``mean-based'' algorithm (we explain precisely what is meant by such an algorithm later) that runs in time $\exp(\tilde{O}(\sqrt{n}))$ and uses $\exp(\tilde{O}(\sqrt{n}))$ traces for any deletion rate $\delta$ that is bounded away from 1 by a constant; this is the prior work that is most relevant to our main positive result.  \cite{HMPW08} also gave a lower bound showing that for any $\delta$ bounded away from 0 by a constant, at least $n^{\Omega({\frac {\log n}{\log \log n}})}$ traces are required for any mean-based algorithm. Since the result of \cite{HMPW08}, several researchers (such as~\cite{Mossel:open}) have raised the question of finding (potentially inefficient) algorithms  which have a better sample complexity; however, no progress had been  made until this work.

One may also ask (as was done in the ``open questions'' of~\cite[Section~7]{Mit09}) for trace reconstruction for more general channels, such as those that allow deletions, insertions, and bit-flips. The only work we are aware of along these lines is that of Andoni~et~al.~\cite{ADHR12}, which gives results for trace reconstruction for \emph{average-case} words in the presence of insertions, deletions, and  substitutions on a tree.

\subsection{Our results}

\begin{theorem} [Deletion channel positive result] \label{thm:main-positive}
There is an algorithm for the trace reconstruction problem which, for any constant $0 < \delta < 1$, uses $\exp(O(n^{1/3}))$ traces and running time.
\end{theorem}

Theorem~\ref{thm:main-positive} significantly improves the running time and sample complexity of the~\cite{HMPW08} algorithm, which is $\exp(\tilde{O}(n^{1/2}))$ for fixed constant~$\delta$. Furthermore, we can actually extend Theorem~\ref{thm:main-positive} to the case of $\delta = o(1)$ or $\delta = 1-o(1)$; see Theorem~\ref{thm:subconstant} below.

The algorithm of Theorem~\ref{thm:main-positive} is a ``mean-based'' algorithm, meaning that it uses only the empirical mean of the trace vectors it receives.  We prove an essentially matching lower bound for such algorithms:

\begin{theorem} [Deletion channel negative result] \label{thm:main-negative}
For any constant $0 < \delta < 1$, every mean-based algorithm must use at least $\exp(\Omega(n^{1/3}))$ traces.
\end{theorem}

As mentioned, we can also treat $\delta = o(1)$ and $\delta = 1-o(1)$:
\begin{theorem} [Deletion channel general matching bounds] \label{thm:subconstant}
    The matching bounds in Theorems~\ref{thm:main-positive} and~\ref{thm:main-negative} extend as follows:
	For $O(\log^3 n)/n \leq \delta \leq 1/2$, the matching bound is $\exp(\Theta(\delta n)^{1/3})$ (and for any smaller~$\delta$ we have a $\poly(n)$ upper bound).  Writing $\rho = 1-\delta$ for the ``retention'' probability, for $O(1/n^{1/2}) \leq \rho \leq 1/2$ the matching bound is $\exp(\Theta(n/\rho)^{1/3})$.
\end{theorem}

For simplicity in the main portion of the paper we consider only the deletion channel and prove the above results.  In Appendix~\ref{app:general} we consider a more general channel that allows for deletions, insertions, and bit-flips, and prove the following result, which extends Theorem~\ref{thm:main-positive} to that more general channel and includes Theorem~\ref{thm:main-positive} as a special case.

\begin{theorem} [General channel positive result] \label{thm:main-positive-general}
Let $\channel$ be the general channel described in Section~\ref{sec:channel-definition} with  deletion probability $\delta = 1-\rho$, insertion probability $\sigma$, and bit-flip probability $\gamma/2$.  Define
\[
r \coloneqq {\frac {\rho + \delta \sigma}{1+\sigma}}.
\]
Then there is an algorithm for $\channel$-channel trace reconstruction using samples and running time bounded by
\[
	\poly(\tfrac{1}{1-\delta}, \tfrac{1}{1-\sigma},\tfrac{1}{1-\gamma}) \cdot \begin{cases}
    	\exp(O(n/r)^{1/3}) & \text{if $C/n^{1/2} \leq r \leq 1/2$,} \\
        \exp(O((1-r) n)^{1/3}) & \text{if $O(\log^3 n)/n \leq 1-r \leq 1/2$.}
    \end{cases}
\]
\end{theorem}

Since some slight technical and notational unwieldiness is incurred by dealing with the more general channel, we defer the proof of Theorem~\ref{thm:main-positive-general} to Appendix~\ref{app:general}; however, we note here that the main core of the proof is unchanged from the deletion-only case.  We additionally note that, as discussed in Appendix~\ref{app:general}, a curious aspect of the upper bound given by Theorem~\ref{thm:main-positive-general} is that having a constant insertion rate can make it possible to perform trace reconstruction in time $\exp(O(n^{1/3}))$ even when the deletion rate is much higher than Theorem~\ref{thm:subconstant} could handle in the absence of insertions.  A possible intuitive explanation for this is that having random insertions could serve to ``smooth out'' worst-case instances that are problematic for a deletion-only model.

\subsection{Independent and concurrent work}
At the time of writing, we have been informed \cite{NP16} that Fedor Nazarov and Yuval Peres have independently obtained results that are substantially similar to Theorems~\ref{thm:main-positive} and~\ref{thm:main-negative}. Also, Elchanan Mossel has informed us \cite{Mossel16} that around 2008, Mark Braverman, Avinatan Hassidim and Elchanan Mossel had independently proven (unpublished) superpolynomial lower bounds for mean-based algorithms.\ignore{ but as it was never written down, the precise bounds are not available to us. }

\subsection{Our techniques}
For simplicity of discussion, we restrict our focus in this section to the question of upper bounding the sample complexity of trace reconstruction for the deletion channel, where every bit gets deleted independently with probability~$\delta$.  (As discussed above, generalizing the results to channels which also allow
for insertions and flips is essentially a technical exercise that does not require substantially new ideas.) As we discuss in Section~\ref{sec:complexity}, an efficient \emph{algorithm} follows easily from a sample complexity upper bound via the observation that the minimization problem whose solution yields a sample complexity upper bound, extends to a slightly larger \emph{convex} set, and thus one can use convex (in fact, linear) programming to get an algorithmic result. Hence the technical meat of the argument lies in upper bounding the sample complexity.

The key enabling idea for our work is to take an analytic view on the combinatorial process defined by the deletion channel. More precisely, consider two distinct strings $x , x' \in \{-1,1\}^n$.  A necessary (and sufficient) condition to upper bound the sample complexity of trace reconstruction is to lower bound the statistical distance between the two distributions of traces of $x$ versus $x'$ (let us write  $\calC (x)$ and $\calC (x')$ to denote these two distributions). Since analyzing the statistical distance $\dtv{\calC ( x )}{ \calC (x') }$ between the distributions $\calC (x)$ and $\calC (x')$ turns out to be a difficult task, we approach it by considering a limited class of statistical tests.

In \cite{HMPW08} the authors consider ``mean-based'' algorithms; such algorithms correspond to statistical tests that only use $1$-bit marginals of the distribution of the received string. More precisely, for any $1 \le j \le n$, consider the quantities $\Pr_{\by \leftarrow \calC (x)} [\by_j=1]$ and $\Pr_{\by' \leftarrow \calC (x')} [\by'_j=1]$. The difference $ \abs{\Pr_{\by \leftarrow \calC (x)} [\by_j=1]- \Pr_{\by' \leftarrow \calC (x')} [\by'_j=1]}$ is a lower bound on $\dtv{\calC ( x )}{ \calC (x') }$.

Let us define the vector $\beta_{x,x'} = (\beta_{x,x'}(1),\dots,\beta_{x,x'}(n)) \in [-1,1]^n$ by
\[\beta_{x,x'}(j)  = \Pr_{\by \leftarrow \calC (x)} [\by'_j=1]- \Pr_{\by' \leftarrow \calC (x')} [\by_j=1].\] In this terminology, giving a sample complexity upper bound on mean-based algorithms correspond to showing a lower bound on\ignore{$\red{\min}_{x \neq x' \in \{-1,1\}^n} \Vert \beta_{x,x'} \Vert_\infty$; up to a factor of $n$, which is negligible in our context, we may consider} $\min_{x \neq x' \in \{-1,1\}^n} \Vert \beta_{x,x'} \Vert_1.$ A central idea in this paper is to analyze $\Vert \beta_{x,x'} \Vert_1$ by studying the $Z$-transform of the vector $\beta_{x,x'}$. More precisely, for $z \in \mathbb{C}$, we consider $\widehat{\beta}_{x,x'}(z) := \sum_{j=1}^n \beta_{x,x'}(j) \cdot z^{j-1}$.
Elementary complex analysis can be used to show that
$$
\sup_{|z|=1}| \widehat{\beta}_{x,x'}(z) | \leq \Vert \beta_{x,x'} \Vert_1 \le \sqrt{n} \cdot\sup_{|z|=1}| \widehat{\beta}_{x,x'}(z) |.
$$
Thus, for our purposes, it suffices to study $\sup_{|z|=1}| \widehat{\beta}_{x,x'}(z)|$. By analyzing the deletion channel and observing that $\widehat{\beta}_{x,x'}(z)$ is a polynomial in $z$, we are able to characterize this supremum as the supremum of a certain polynomial (induced by $x$ and $x'$) on a certain disk in the complex plane.  Thus giving a sample complexity upper bound amounts to lower bounding $\sup_{|z|=1}| \widehat{\beta}_{x,x'}(z)|$ across all polynomials $\widehat{\beta}_{x,x'}$ induced by distinct $x,x' \in \{-1,1\}^n$ (essentially, across a class of polynomials closely related to \emph{Littlewood polynomials: those polynomials with all coefficients in $\{-1,0,1\}$)}. The technical heart of our sample complexity upper bound is in establishing such a lower bound. Finally, similar ideas and arguments are used to lower bound the sample complexity of mean-based algorithms, by upper bounding $\sup_{|z|=1}| \widehat{\beta}_{x,x'}(z)|$ across all polynomials $\widehat{\beta}_{x,x'}$ induced by distinct $x,x' \in \{-1,1\}^n$.

\section{Preliminaries and terminology} \label{sec:prelim}

Throughout this paper we will use two slightly nonstandard notational conventions.  Bits will be written as $\bits$ rather than $\zo$, and strings will be indexed starting from~$0$ rather than~$1$.  Thus the \emph{source string} will be denoted $x = (x_0, x_1, \dots, x_{n-1}) \in \bn$; this is the unknown string that the reconstruction algorithm is trying to recover.

We will write $\channel$ for the \emph{channel} through which~$x$ is transmitted.  In the main body of the paper our main focus will be on the \emph{deletion channel} $\channel = \Del_\delta$, in which each bit of $x$ is independently $\delta$eleted with probability~$\delta < 1$.  We will also often consider $\rho = 1-\delta > 0$, the $\rho$etention probability of each coordinate.  In Appendix~\ref{app:general} we will see that a more general channel that also involves \emph{insertions} and \emph{bit-flips} can be handled in a similar way.  

We will use boldface to denote random variables.  We typically write $\by \leftarrow \channel(x)$ to denote that $\by = (\by_0, \by_1,  \dots, \by_{\boldn - 1})$ is a random \emph{trace} (or \emph{received string} or \emph{sample}),  obtained by passing $x$ through the channel~$\calC$.  Notice the slight inconvenience that the length of $\by$ is a random variable (for the deletion channel this length is always between 0 and $n$); we denote this length by~$\boldn$.

We define a \emph{trace reconstruction algorithm for channel $\channel$} to be an algorithm with the following property:  for any unknown source string $x \in \bn$, when given access to independent strings $\by^{(1)}, \by^{(2)}, \dots$ each distributed according to $\channel(x)$, it outputs $x$ with probability at least (say) $99\%$.  The \emph{sample complexity} of the trace reconstruction algorithm is the number of draws from $\channel(x)$ that it uses (in the worst case across all $x \in \bn$ and all draws from $\channel(x)$).  We are also  interested in the algorithm's (worst-case) running time.

As mentioned earlier we will use basic complex analysis.  The following notation will be useful:
\begin{notation}
We write $D_r(c)$ for the closed complex disk of radius $r$ centered at $c$; i.e., $\{z \in \C : |z - c| \leq r\}$.  We write $\bdry D_r(c)$ for the boundary of this disk; thus, e.g., $\bdry D_1(0) = \{z \in \C  : |z| = 1\}$ is the complex unit circle.
\end{notation}

\section{Mean traces} \label{sec:mean}

\ignore{\rnote{Everything before the Appendix is now just about the deletion channel, hence no big-$N$s; this is for as-simple-as-possible definitions/setup, and to quarantine all the messiness (idealized vs nonidealized mean trace, etc) to the appendix.  This leads to a little reduplication in the appendix, but I think it's okay.}}
We now come to a key definition, that of the \emph{mean trace}.  For now we restrict our focus to $\channel$ being the deletion channel $\Del_\delta$ (we consider a more general channel in Appendix~\ref{app:general}).

Although a random trace $\by \leftarrow \Del_\delta(x)$ does not have a fixed length, we can simply define the mean trace of a source string~$x \in \bn$ to be
\begin{equation} \label{eq:meantrace}
\mu_{\Del_\delta}(x) = \E_{\by \leftarrow \Del_\delta(x)}[\by'] \in [-1,1]^n,
\end{equation}
where $\by'$ is $\by$ padded with zeros so as to be of length exactly $n$.  Here ``$0$'' has a natural interpretation as a ``uniformly random bit'' (indeed, a trace reconstruction algorithm could always pad deletion-channel traces with random bits by itself, and this would not change the definition of the mean trace~$\mu_{\Del_\delta}(x)$).


The following is immediate:
\begin{proposition} \label{prop:its-linear}
Viewing the domain of $\mu_{\Del_\delta}$ as the real vector space $\R^n$, $\mu_{\Del_\delta}(x)$ is a \mbox{(real-)linear} function of~$x$; that is, each $\mu_{\Del_\delta}(x)_j$ can be written as $\sum_i a_{i,j} x_i$ for some constants $a_{i,j} \in \R$.\ignore{\onote{The main reason we want to ``know'' the parameters of the channel is to ``know'' these constants.}}
\end{proposition}

\subsection{The mean-based (deletion-channel) trace reconstruction model}

One of the most basic things that a trace reconstruction algorithm can do is calculate an empirical estimate of the mean trace.  A simple Chernoff/union bound shows that, with  \ignore{\poly(N/\eps) =}$\poly(n/\eps)$ samples and time, an algorithm can compute an estimator $\wh{\mu}_{\Del_\delta}(x) \in [-1,1]^n$ satisfying ${\|\wh{\mu}_{\Del_\delta}(x) - \mu_{\Del_\delta}(x)\|_1 \leq \eps}$
with very high probability.   The algorithm might then proceed to base its reconstruction solely on $\wh{\mu}_{\Del_\delta}(x)$, without relying on further traces.  We call such algorithms ``mean-based trace reconstruction algorithms''\ignore{, as they do not use estimates of pairwise $\E[\by_i \by_j]$ moments or higher moments.} (Holenstein et~al.~\cite{HMPW08} called them algorithms based on ``summary statistics'').  We give a formal definition:

\begin{definition} \label{def:meanbased}
    An algorithm in the \emph{mean-based (deletion-channel) trace reconstruction model} works as follows.  Given an unknown source string $x \in \{-1,1\}^n$, the algorithm first specifies a parameter~$T \in \N$.  The algorithm is then given an estimate $\wh{\mu}_{\Del_\delta}(x) \in [-1,+1]^n$ of the mean trace satisfying
    \begin{equation}    \label{eqn:mean-estimate2}
        \|\wh{\mu}_{\Del_\delta}(x) - \mu_{\Del_\delta}(x)\|_1 \leq 1/T.
    \end{equation}
    We define the ``cost'' of this portion of the algorithm to be~$T$.   Having been given~$\wh{\mu}_{\Del_\delta}(x)$, the algorithm has no further access to~$x$, but may do further ``postprocessing'' computation involving $\wh{\mu}_{\Del_\delta}(x)$.  The algorithm should end by outputting~$x$.\ignore{\onote{I was a bit cagey here about whether the algorithm has to be deterministic.  Of course: a) our algorithm is; b) our lower bound would apply even to an algorithm that was randomized after getting $\wh{\mu}_{\Del_\delta}(x)$.  I didn't feel like it was worth getting into these details.}}
\end{definition}
From the above discussion, we see that an algorithm in the mean-based trace reconstruction model with cost $T_1$ and postprocessing time $T_2$ may be converted into a normal trace reconstruction algorithm using $\poly(n,T_1)$ samples and $\poly(n,T_1) + T_2$ time.\ignore{  (As an aside, since we will be studying algorithms with cost $T \ll 2^n$, there is no real difference between getting an estimate of~$\mu_{\Del_\delta}(x)$ or of $\mu^{\ideal}_{\Del_\delta}(x)$.)}

\ignore{
A closely related notion is the \emph{$m$-sample empirical mean trace of $x$}, denoted $\widehat{\bmu}_{m,\delta}(x)$.  This is a random variable obtained by drawing $m$ independent traces from $\Del_\delta(x)$ and outputting the mean of those $m$ vectors.  When the deletion rate $\delta$ and the source string $x$ are clear from context we sometimes simply write $\mu$ for $\mu_\delta(x)$ and $\widehat{\bmu}_{m}$ for $\widehat{\bmu}_{m,\delta}(x).$

We define an \emph{$m$-sample mean-based algorithm} as an algorithm for the trace reconstruction problem which works in the following way:  for any source string $x$ it receives a draw from the $m$-sample empirical mean trace $\widehat{\bmu}_{m,\delta}(x)$, does some computation, and outputs a hypothesis string (in $\zo^n$) which must be equal to $x$ with probability at least (say) 99/100.\ignore{\rnote{Checking that we are okay with using this as the definition of a mean-based algorithm.  An alternative would be to go all 1-RFA on it and talk about algorithms that can only look at one bit, etc, but I'm not sure there is a point in doing that.}}\ignore{An alternate definition would be to define an $m$-sample mean-based algorithm as an algorithm which can draw independent traces $\by^{(1)},\by^{(2)},\dots \leftarrow \Del_\delta(x)$ but for each received trace $\by^{(i)}$ it can only specify a coordinate $j \in [0,n-1]$ and be shown the bit $\by^{(i)}_j$, after which $\by^{(i)}$ is discarded and rendered subsequently unavailable.  It is clear that given an $m$-sample mean-based algorithm of the first type, there is an equivalent $nm$-sample mean-based algorithm of the second type.}  It is clear that any $m$-sample mean-based  algorithm immediately gives a trace reconstruction algorithm with sample complexity $m$.
}

\ignore{





}
\ignore{
For a vector $v \in \R^n$ we write $\|v\|$ to denote the 2-norm $\sqrt{\sum_j v_j^2}$ and $\|v\|_1$ to denote the 1-norm $\sum_j |v_j|.$ We write $\Bin(k,\tau)$ to denote a Binomially distributed random variable with parameters $k \in \N$ and $0 \leq \tau \leq 1$, so $\Pr[\Bin(k,\tau)=j] = {k \choose j} \tau^j (1-\tau)^{k-j}.$  We write $D_r(z)$ to denote the disk of radius $r$ centered at $z$ in the complex plane and $\partial D_r(z)$ to denote the circle which is its boundary.
}

\subsection{The complexity of mean-based (deletion-channel) trace reconstruction} \label{sec:complexity}

As discussed in~\cite{HMPW08}, the sample complexity of mean-based  trace reconstruction is essentially determined by the minimum distance between the mean traces $\mu_{\Del_\delta}(x)$ and $\mu_{\Del_\delta}(x')$ of two distinct source strings $x, x' \in \bn$.  Furthermore, one can get an upper bound on the \emph{time} complexity of mean-based trace reconstruction if a certain ``fractional relaxation'' of this minimum mean trace distance is large.  We state these observations from~\cite{HMPW08} here, using slightly different notation.

\begin{definition}      \label{def:gap}
    Given $n$ and $0 \leq \delta < 1$, we define:
    \begin{align*}
        \gap_{\Del_\delta}(n) \coloneqq  \min_{\substack{\ x, x' \in \bn\ \\ x \neq x'}} \|\mu_{\Del_\delta}(x) - \mu_{\Del_\delta}(x')\|_1  &=2 \min_{\substack{b \in \{-1,0,+1\}^n \\ b \neq 0}} \|\mu_{\Del_\delta}(b)\|_1;\\
        \gap^{\chill}_{\Del_\delta}(n) \coloneqq \min_{\ 0 \leq i < n\ } \min_{\substack {x, x' \in [-1,+1]^n \\ x_j = x'_j \in \bits \forall j < i \\ x_i = -x'_i \in \bits}} \|\mu_{\Del_\delta}(x) - \mu_{\Del_\delta}(x')\|_1
          &= 2\min_{\ d \in [n]\ } \min_{b \in \{0\}^{d-1} \times \{1 \} \times [-1,+1]^{n-d}} \|\mu_{\Del_\delta}(b)\|_1.
    \end{align*}
    In both cases, the equality on the right uses Proposition~\ref{prop:its-linear}.
\end{definition}
It's easy to see that in the mean-based trace reconstruction model, it is information-theoretically possible for an algorithm to succeed if and only if its  cost~$T$ exceeds $2/\gap_{\Del_\delta}(n)$.  Thus characterizing the sample complexity of mean-based trace reconstruction essentially amounts to analyzing~$\gap_{\Del_\delta}(n)$.  For example, to establish our lower bound Theorem~\ref{thm:main-negative}, it suffices to prove that the $\gap_{\Del_\delta}(n) \leq \exp(-\Omega(n^{1/3}))$ for constant $0 < \delta < 1$.

Furthermore, as observed in~\cite{HMPW08}, given an $\gap^{\chill}_{\Del_\delta}(n)/4$-accurate estimate of $\mu_{\Del_\delta}(x)$, as well as the ability to compute the linear function $\mu_{\Del_\delta}(x')$ for any $x' \in [-1,+1]^n$ (or even estimate it to $\gap^{\chill}_{\Del_\delta}(n)/4$-accuracy), one can recover~$x$ exactly in $\poly(n,\log(1/\gap^{\chill}_{\Del_\delta}(n)))$ time by solving a sequence of~$n$ linear programs.\footnote{If the algorithm ``knows'' $\delta$ it can efficiently compute $\mu_{\Del_\delta}(x')$ exactly.  But even if it doesn't ``know'' $\delta$, it can estimate $\delta$ to sufficient accuracy so that $\mu_{\Del_\delta}(x')$ can be estimated to the necessary accuracy, with no significant algorithmic slowdown.} Thus to establish our Theorem~\ref{thm:main-positive}, it suffices to prove that $\gap^{\chill}_{\Del_\delta}(n) \geq \exp(-O(n^{1/3}))$ for constant $0 < \delta < 1$.

\subsection{Reduction to complex analysis} \label{sec:reduc}

Our next important definition is of a polynomial that encodes the components of $\mu_\channel(x)$ in its coefficients --- kind of a generating function for the channel.  We think of its parameter~$z$ as a complex number.
\begin{definition}      \label{def:P}
    Given $x \in \{-1,1\}^n$ and $0 \leq \delta < 1$, we define the \emph{deletion-channel polynomial}
    \[
            P_{\Del_\delta,x}(z) = \sum_{j < n} \mu_{\Del_\delta}(x)_j \cdot z^j,
    \]
    a polynomial of degree less than~$n$.  We extend this definition to $x \in [-1,+1]^n$ using the linearity of $\mu_{\Del_\delta}$.
\end{definition}

We now make the step to elementary complex analysis, by relating the size of a mean trace difference $\mu_{\Del_\delta}(b)$ to the maximum modulus of $P_{\Del_\delta,b}(z)$ on the unit complex circle (or equivalently, the unit complex disk, by the Maximum Modulus Principle):
\begin{proposition} \label{prop:cx}
For any $b \in [-1,1]^n$, we have
\[
         \max_{z \in \bdry D_1(0)} \abs*{P_{\Del_\delta,b}(z)} \leq \|\mu_{\Del_\delta}(b)\|_1 \leq \sqrt{n} \max_{z \in \bdry D_1(0)} \abs*{P_{\Del_\delta,b}(z)}.
\]
\end{proposition}
\begin{proof}
    Recall that $\mu_{\Del_\delta}(b)$ is the length-$n$ vector of coefficients for the polynomial $P_{\Del_\delta,b}(z)$.  The lower bound above is immediate from the triangle inequality.  For the upper bound, we use
    \[
        \|\mu_{\Del_\delta}(b)\|_1^2 \leq n \|\mu_{\Del_\delta}(b)\|_2^2 = n  \Avg_{z \in \bdry D_1(0)} \abs*{P_{{\Del_\delta},b}(z)}^2 \leq n \parens*{\max_{z \in \bdry D_1(0)} \abs*{P_{{\Del_\delta},b}(z)}}^2.
    \]
    Here the first inequality is Cauchy--Schwarz, the equality is an elementary fact about complex polynomials (or Fourier series), and the final inequality is obvious.
\end{proof}
\ignore{
\begin{fact} \label{fact:length-modulus}
For any real vector $v=(v_0,\dots,v_{n-1})$, the squared Euclidean length of $v$ is equal to the average of the squared modulus of $G(v;Z)$ averaged around the unit circle of the complex plane.
\end{fact}
\begin{proof}
\rnote{can/should we cite something here instead of writing this proof?  I guess it's short...}The average squared modulus of $G(v;Z)$ is
\begin{align*}
\Avg_{|z|=1} \left| \sum_{j=0}^{n-1} v_j z^j \right|^2
&= {\frac 1 {2\pi}} \int_{0}^{2 \pi}  \left| \sum_{j=0}^{n-1} v_j e^{ijt} \right|^2 dt
= {\frac 1 {2\pi}} \int_{0}^{2 \pi}  \left| \sum_{j=0}^{n-1} v_j \cos jt + i \sum_{j=0}^{n-1} v_j \sin jt \right|^2 dt\\
&= {\frac 1 {2\pi}} \int_{0}^{2 \pi}  \left(\sum_{j=0}^{n-1} v_j \cos jt\right)^2 + \left(\sum_{j=0}^{n-1} v_j \sin jt \right)^2 dt\\
&= {\frac 1 {2\pi}} \int_{0}^{2 \pi}  \sum_{j=0}^{n-1} v_j^2 (\cos^2 jt + \sin^2 jt)
+ \sum_{j_1 \neq j_2} v_{j_1} v_{j_2} ( \cos j_1 t \cos j_2 t + \sin j_1 t \sin j_2 t)
 dt\\
 &= \sum_{j=0}^{n-1} v_j^2 = \|v\|^2.
\end{align*}
\end{proof}
}

Let us reconsider Definition~\ref{def:gap}. As a factor of $\sqrt{n}$ is negligible compared to the bounds we will prove (which are of the shape $\exp(-\Theta(n^{1/3}))$, we may as well analyze $\max_{z \in \bdry D_1(0)} \abs*{P_{{\Del_\delta},b}(z)}$ rather than $\|\mu_{\Del_\delta}(b)\|_1$ in the definition of $\gap_{\Del_\delta}(n)$ and $\gap^{\chill}_{\Del_\delta}(n)$.  We therefore take a closer look at the deletion-channel polynomial.

\ignore{
\section{Ryan's Oct.\ 6 note about polynomials}

OK, here's how I look at it.  Let me warm up with insertions only.  Model is: you proceed thru each transmission character.  On a given character, you repeatedly: ``transmit a random bit with probability $1-\tau$, transmit the $\tau$rue bit with probability $\tau$''.  You do that until you transmit the true bit.  Then you go on to the next one.  Equivalently: For each true bit, you first transmit Geometric$(\tau) - 1$ random bits, then you transmit the true bit.

Let me analyze this in slightly more generality than we need, so as to be able to talk deletions later.  I will use the notation $i \rightsquigarrow j$ to mean that true bit $x_i$ was eventually transmitted at ``time'' $j$ (so that it becomes $y_j$ in the received string).  In the insertion model this always makes sense.  In the deletion model, sometimes $x_i$ is straight-up deleted.  If you want, write $i \rightsquigarrow \bot$ for this outcome.

OK, sorry about this Anindya, but I think it's pretty helpful to denote bits as $\pm 1$, rather than $\{0,1\}$.  So let's do this.  The key reason is that if you have some $j$ such that $\not \exists i (i \rightsquigarrow j)$, then $\by_j$ has the property that it's uniformly random and hence $\E[\by_j] = 0$. So we have
\[
	\sum_{j} \E[\by_j] z^j = \sum_{j} \sum_{i} \Pr[i \rightsquigarrow j]x_i z^j,
\]
where crucially we used that the conditional expectation of $\by_j$ is $0$ if no $i$ squiggly-arrows to~$j$. Let me add that this convention is pretty convenient in the deletion-channel model too, for $j$'s that end up ``past the length of the source string''.  Like, we can imagine that once the deletion channel is finished transmitting it  just starts transmitting random bits, and we can have the sum over~$j$ extend to infinity (not just to~$n$) if we want.  Anyway, back to the last displayed equation, interchanging the sums it's
\[
	\sum_i x_i \sum_j \Pr[i \rightsquigarrow j] z^j = \sum_i x_i \text{``}\E\text{''}[z^{\bJ_i}].
\]
Here $\bJ_i$ is the random variable denoting the eventual position at which $x_i$ was transmitted. (Of course, $\bJ_1, \bJ_2, \dots$ are not independent random variables.)  I also put the expectation in quotation marks for the sake of the deletion channel, because if $i \rightsquigarrow \bot$ then $\bJ_i$ is undefined (and you should just drop this part in the summation).  However, it's easy to see, in analyzing the inner summation given~$i$, that if you just multiply by $\rho$ you get that it's equivalent to conditioning on $i$ not being deleted, and then $\bJ_i$ is a proper random variable and you can drop the quotation marks.

OK, but again, coming back to the insertion channel.  We now just have to compute $\E[z^{\bJ_i}]$.  In the insertion channel, $\bJ_i$ is distributed as the sum of $i$ independent Geometric$(\tau)$ random variables. So
\[
	\E[z^{\bJ_i}] = \E[z^{\bG_1 + \cdots + \bG_i}] = \E[z^{\bG}]^i,
\]
where $\bG, \bG_1, \bG_2, \dots$ denote iid Geom$(\tau)$'s.  I dunno if there's a super-cute way to compute the $z$-transform of a geometric, but anyway one can just calculate:
\[
	\E[z^{\bG}] = \sum_{r=1}^\infty (1-\tau)^{r-1} \tau z^r = \frac{\tau z}{1 - (1-\tau)z}.
\]
Let me write $\mathrm{Ins}_\tau(z)$ for that function.  So finally,
\[
	\sum_j \E[\by_j] z^j = \sum_i x_i \mathrm{Ins}_\tau(z)^i.
\]
And note that $\mathrm{Ins}_\tau(z)$ traces out the complex circle centered at $\frac{1-\tau}{2-\tau}$, passing through~$1$.  So it's always a nice circle containing the origin, even for $\tau$ very close to~$0$.

Going back up a bit, if you want to do deletion channel, you want to compute $\E[z^{\bJ_i}]$, where recall $\bJ_i$ denotes the location that $x_i$ ends up in, conditioned on it not being deleted.  Here $\bJ_i$ is distributed as Binomial$(i-1, \rho)+1$.  (I'm returning to traditional indexing $1, 2, \dots$ for strings.).  Here it's easy to compute the $z$-transform, b/c you can write the Binomial as $1 + \bB_1 + \cdots + \bB_{i-1}$, where $\bB_j$'s are iid Bernoulli$(\rho$)'s.  So\dots well, etc.  Let me, in the next section, write how I think we could write it properly/carefully.
}

\section{The deletion-channel polynomial} \label{sec:channel-polynomial-delete}

In this section we compute the deletion-channel polynomial.  When the deletion channel is applied to some source string~$x$, each bit $x_i$ is either deleted with probability~$\delta$ or else is transmitted at some position~$j \leq i$ in the received string~$\by$.  Let us introduce (non-independent) random variables $\bJ_0, \dots, \bJ_{n-1}$, where $\bJ_i = \bot$ if $x_i$ is deleted and otherwise $\bJ_i$ is the position in~$\by$ at which~$x_i$ is transmitted.
We thus have
\[
	P_{\Del_\delta,x}(z) = \sum_{j < n} \E_{\by \leftarrow \channel(x)}[\by_j] \cdot z^j = \sum_{j<n} z^j \cdot  \sum_{i<n} \Pr[\bJ_i = j] x_i = \sum_{i<n} x_i \cdot \sum_{j<n} \Pr[\bJ_i = j] z^j = \sum_{i<n} x_i  \cdot \text{``}\E\text{''}[z^{\bJ_i}].
\]
Here we put the expectation $\E$ in quotation marks because the expression should count~$0$ whenever $\bJ_i = \bot$.  Observing that $\Pr[\bJ_i \neq \bot]$ equals the retention probability $\rho = 1-\delta$, if we define the conditional random variable
\[
 \wt{\bJ}_i = (\bJ_i \mid \bJ_i \neq \bot)
\]
(so $\wt{\bJ}_i$ is an $\N$-valued random variable), then we have
\begin{equation}    \label{eqn:P-formula}
    P_{\Del_\delta,x}(z) = \rho  \sum_{i<n} x_i  \cdot \E[z^{\wt{\bJ}_i}].
\end{equation}
Observing that $\wt{\bJ}_i$ is distributed as $\text{Binomial}(i,\rho)$, and letting $\bB_1, \dots, \bB_i$ denote independent Bernoulli random variables with ``success'' probability~$\rho$, we easily compute
\[
    \E[z^{\wt{\bJ}_i}] = \E[z^{\bB_1 + \cdots + \bB_i}] = \E[z^{\bB_1}]^i = ((1-\rho) + \rho z)^i.
\]
Denoting
\[
    w = 1-\rho + \rho z,
\]
we conclude that
\[
    P_{\Del_\delta,x}(z) = \rho \sum_{i < n} x_i w^i.
\]

As $z$ ranges over the unit circle $\partial D_1(0),$ $w$ ranges over the radius-$\rho$ circle $\partial D_{\rho}(1-\rho).$  Recalling Definition~\ref{def:gap} and Proposition~\ref{prop:cx}, we are led to consider the following two quantities for $0 < \rho < 1$ (note that by the Maximum Modulus Principle, these quantities are unchanged whether the $\max$ is over $D_\rho(1-\rho)$ or $\partial D_\rho(1-\rho)$):
\begin{align*}
    \kappa_{\Lit}(\rho,n) &= \min \braces*{\max_{w \in D_\rho(1-\rho)} \abs*{P(w)} : P(w) = b_0 + b_1 w  + \cdots + b_{n-1} w^{n-1},\ b_i \in \{0, \pm 1\} \text{ not all~$0$}},\\
   \kappa_{\bounded}^{\chill}(\rho,d) &= \min \braces*{\max_{w \in D_\rho(1-\rho)} \abs*{P(w)} : P(w) = w^d + b_{d+1} w^{d+1} + \cdots + b_N w^N,\ N \geq d, \ b_i \in D_1(0)}.
\end{align*}
By the Maximum Modulus Principle, both $\kappa_\Lit(\rho,n)$ and $\kappa^{\chill}_{\bounded}(\rho,d)$ are nondecreasing functions of~$0<\rho<1$.  It's also easy to see that both are nonincreasing functions of their second argument for all $0<\rho<1$ (for $\kappa_{\bounded}^{\chill}(\rho,d)$, consider replacing $P(w)$ by $wP(w)$) and observe that $|wP(w)| \leq |P(w)|$ for all $w \in D_\rho(1-\rho)$). It thus follows that
\[
     \kappa_{\bounded}^{\chill}(\rho,d) \leq \kappa_{\Lit}(\rho,d).
\]

Our main technical theorems are the following:
\begin{theorem}                                    \label{thm:Lit-chill-lower}
There is a universal constant $C \geq 1$ such that:
    \begin{align*}
    \text{for $1/d \leq \delta \leq 1/2$,}& & \kappa_{\bounded}^{\chill}(1-\delta,d) &\geq \exp(-C (\delta d)^{1/3}); \\
        \text{for $1/d^{1/2} \leq \rho \leq 1/2$,}& & \quad \kappa_{\bounded}^{\chill}(\rho,d) &\geq \exp(-C(d/\rho)^{1/3}).
   	\end{align*}
\ignore{
$\displaystyle \kappa_{\bounded}^{\chill}(\rho,d) \geq \exp(-C(d/\rho)^{1/3})$ provided that $1/d^{1/2} \leq \rho < 1$.
}
\ignore{
    XXX OLD, STILL TRUE VERSION:
    Provided $1/d^{1/2} \leq r \leq .1$ it holds that $\displaystyle \kappa^{\chill}_{\Lit}(r,d) \geq (4040)^{-(d/r)^{1/3}}$.
}
\end{theorem}
\begin{theorem}                                     \label{thm:Lit-upper}
	There is a universal constant $C \geq 1$ such that:
    \begin{align*}
    \text{for $C(\log^3 n)/n \leq \delta \leq 1/2$,}& & \kappa_{\Lit}(1-\delta,n) &\leq \exp(-\Omega(\delta n)^{1/3}); \\
        \text{for $C/n^{1/2} \leq \rho \leq 1/2$,}& & \quad \kappa_{\Lit}(\rho,n) &\leq \exp(-\Omega(n/\rho)^{1/3}).
   	\end{align*}
\end{theorem}

By Definition~\ref{def:gap}, Proposition~\ref{prop:cx}, and the discussion at the end of Section~\ref{sec:complexity}, we have that Theorem~\ref{thm:Lit-upper} implies both Theorem~\ref{thm:main-negative} and the more general sample complexity lower bound in Theorem~\ref{thm:subconstant}. Regarding the algorithmic upper bounds in Theorems~\ref{thm:main-positive} and~\ref{thm:subconstant}, again from Definition~\ref{def:gap} and Proposition~\ref{prop:cx} we get that
\begin{align*}
    \gap^{\chill}_{\Del_\delta}(n) &\geq 2{\rho}\cdot  \min_{0 \leq d < n} \braces*{\max_{w \in D_\rho(1-\rho)} \abs*{P(w)} : P(w) = w^d + b_{d+1} w^{d+1} + \cdots + b_{n-1} w^{n-1}, \ b_i \in [-1,+1]} \\
    &\geq 2{\rho} \cdot \min_{0 \leq d < n} \kappa_{\bounded}^{\chill}(\rho,d) \geq 2{\rho} \cdot \kappa_{\bounded}^{\chill}(\rho, n).
\end{align*}
Thus the upper bounds Theorems~\ref{thm:main-positive} and~\ref{thm:subconstant} likewise follow from Theorem~\ref{thm:Lit-chill-lower} and the discussion at the end of Section~\ref{sec:complexity}. (Note that if $\delta \leq O(\log^3 n)/n$, we can always pay the bound for the larger value $\delta = \Theta(\log^3 n)/n)$, which is $\poly(n)$.)

\section{Proof of Theorem~\ref{thm:Lit-chill-lower}}

We will need the following:
\begin{theorem}   \label{thm:be}
    (\cite{BE97}, Corollary~3.2, $M = 1$ case.)  Let $Q(w)$ be a polynomial with constant coefficient~$1$ and all other coefficients bounded by~$1$ in modulus.  Fix any $0 < \theta \leq \pi$, and let $A$ be the arc $\{e^{i t} : -\theta \leq t \leq \theta\}$.  Then $\sup_{w \in A} |Q(w)| \geq \exp(-C_1/\theta)$ for some universal constant~$C_1$.
\end{theorem}
We remark that for any $0 < r < 1$, Theorem~\ref{thm:be} holds for the arc $A = \{r e^{i t} : -\theta \leq t \leq \theta\}$ with no change in the constant~$C_1$. This is immediate by applying the theorem to $\wt{Q}(w) = Q(rw)$.

\begin{proof}[Proof of Theorem~\ref{thm:Lit-chill-lower}.]
Fix $d \geq 2$ (else the hypotheses are vacuous) and  $\delta + \rho = 1$.  We call Case~I when $1/d \leq \delta < 1/2$, and we call Case~II when $1/d^{1/2} \leq \rho \leq 1/2$.  Select
\[
	\theta = \begin{cases}
    			\frac{1}{2(\delta d)^{1/3}} & \text{in Case~I,} \\
                \left(\frac{\rho}{d}\right)^{1/3} \vphantom{x^{x^{x^{x^{x^{x}}}}}}& \text{in Case~II.}
    	      \end{cases}
\]
In Case~I we have $\theta \leq 1/2$, and in Case~II we have $\theta \leq \rho \leq 1/2$.

Let $P(w) = w^d \cdot Q(w)$, where $Q(w)$ is a polynomial with constant coefficient~$1$ and all other coefficients bounded by~$1$ in modulus.  We need to show
\begin{equation} \label{eqn:finisher}
	\max_{w \in D_\rho(\delta)} |P(w)| \geq \begin{cases}
    			\exp(-C(\delta d)^{1/3}) & \text{in Case~I,} \\
                \exp(-C(d/\rho)^{1/3}) & \text{in Case~II.}
   \end{cases}
\end{equation}
In Case~I, the ray $\{re^{i\theta} : r > 0\}$ intersects $\bdry D_\rho(\delta)$ at a unique point, call it~$w_0$.  In Case~II, the same ray intersects $D_\rho(\delta)$ twice (this uses $\theta \leq \rho$); call the point of larger modulus~$w_0$.  In either case, consider the triangle formed in the complex plane by the points $0$, $\delta$, and $w_0$; it has some acute angle~$\alpha$ at~$w_0$ and an angle of~$\theta$ at~$0$. By the Law of Sines,
\begin{multline*}
	\frac{\rho}{\sin \theta} = \frac{\delta}{\sin \alpha} = \frac{|w_0|}{\sin(\pi - \theta - \alpha)} = \frac{|w_0|}{\sin(\theta + \alpha)} = \frac{|w_0|}{\sin\theta\cos \alpha + \sin\alpha \cos\theta} \\
    \implies |w_0| = \delta \cos \theta + \rho \cos \alpha = \delta \cos \theta + \rho \sqrt{1-(\tfrac{\delta}{\rho})^2 \sin^2 \theta} \geq \delta (1-\theta^2) + \rho (1-(\tfrac{\delta}{\rho})^2 \theta^2) = 1 - \tfrac{\delta}{\rho}\theta^2.
\end{multline*}
(The last inequality used $\theta \leq \rho$ in Case~II.) Writing $r_0 = |w_0|$,  Theorem~\ref{thm:be} (and the subsequent remark) implies that
\begin{equation}    \label{eqn:proveth-me}
    \max_{w \in A} |Q(w)| \geq \exp(-C_1/\theta) \quad \text{for } A = \{r_0 e^{it} : -\theta \leq t \leq \theta\} \subset D_\rho(\delta).
\end{equation}
Thus
\[
    \max_{w \in D_\rho(\delta)} |P(w)| \geq \max_{w \in A} |P(w)| \geq r_0^d \cdot \exp(-C_1/\theta) \geq (1-(\delta/\rho)\theta^2)^d \cdot  \exp(-C_1/\theta) \geq \exp(-2(\delta/\rho)\theta^2d - C_1/\theta)
\]
(the last inequality again using $\theta \leq \rho$ in Case~II). Substituting in the value of~$\theta$ yields~\eqref{eqn:finisher}. \ignore{
Let $1/d^{1/2} \leq \rho < 1$ and let $P(w) = w^dQ(w)$, where $Q(w)$ is a polynomial with constant coefficient 1 and all other coefficients bounded by 1 in modulus.  We shall show, for some universal~$C$, that
\begin{equation}    \label{eqn:prove-me-1}
    \max_{w \in A} \abs*{P(w)} \geq \exp(-C(d/\rho)^{1/3}),
\end{equation}
where $A$ is an arc lying entirely in $D_\rho(1-\rho)$; this clearly gives the theorem.

Set $\theta = (\rho/d)^{1/3}$. Since $\rho \geq 1/d^{1/2}$ we have $\theta \leq \rho$; this inequality implies that the ray $\{Re^{i\theta} : R> 0\}$ intersects  $\bdry D_\rho(1-\rho)$ (once if $\rho \geq 1/2$, twice if $\rho < 1/2$).    Let $w_0$ be the intersection point of largest modulus,~$R_0$.  It is simple to see that $1 - R_0$ is proportional to~$\theta^2/\rho$; it follows that $R_0 \geq \exp(-C_2 \theta^2/\rho)$ for some universal~$C$.  Letting $A = \{R_0e^{it} : -\theta \leq t \leq \theta\} \subset D_\rho(1-\rho)$, we get
\begin{equation}    \label{eqn:combine1}
    w \in A \implies |w|^d = R_0^d \geq \exp(-C_2 \theta^2/\rho)^d\geq \exp(-C_2(d/\rho)^{1/3}).
\end{equation}
  At the same time, Theorem~\ref{thm:be} and our remark after it tell us that
\begin{equation}    \label{eqn:combine2}
    \max_{w \in A} |Q(w)| \geq \exp(-C_1/\theta) = \exp(-C_1 (d/\rho)^{1/3}).
\end{equation}
Combining~\eqref{eqn:combine1} and \eqref{eqn:combine2} establishes~\eqref{eqn:prove-me-1}.
%
}
\end{proof}

\subsection{An improved version}        \label{sec:improvement}
Although we don't need it for our application, we can actually provide a stronger version of the results in the previous section that is also self-contained --- i.e., it does not rely on Borwein and Erd\'elyi's Theorem~\ref{thm:be}.  We used that theorem to establish~\eqref{eqn:proveth-me}; but more strongly than~\eqref{eqn:proveth-me}, we can show there exists an arc $A \subset D_\rho(\delta)$ such that
\[
    \GM_{w \in A} |Q(w)| \geq \exp(-O(1/\theta)),
\]
where the left-hand side here denotes the \emph{geometric mean} of $|Q|$ along~$A$.  (Of course, this is at most the max of $|Q|$ along~$A$.)  To keep the parameters simpler, we will assume $\rho \leq 1/3$ (this is the more interesting parameter regime anyway, and it is sufficient to yield our Theorem~\ref{thm:main-positive}).  Our alternate arc $A$ will be
\[
    A = \{1/3 + r e^{it} : -\theta \leq t \leq \theta\},
\]
where $0 < r < 2/3$ is the larger real radius such that $1/3 + r e^{\pm i \theta} \in \bdry D_\rho(\delta)$.  We remark that still $A \subset D_\rho(\delta)$, by virtue of $\theta \leq \rho \leq 1/3$, and it is not hard to show that the the endpoint of~$A$, call it $w' = 1/3 + re^{i\theta} \in \bdry D_\rho(\delta)$, again satisfies $|w'| \geq 1 - \Omega(\tfrac{\delta}{\rho} \theta^2)$.  Thus instead of using Theorem~\ref{thm:be} as a black box, we could have completed our proof of Theorem~\ref{thm:Lit-chill-lower} using the following:
\begin{theorem} \label{thm:BE-like}
    Let $Q(w)$ be a polynomial with constant coefficient~$1$ and all other coefficients in~$D_1(0)$.  Fix any $0 < \theta \leq \pi$, $0 \leq r \leq 2/3$, and let $A$ be the arc $\{1/3 + re^{i t} : -\theta \leq t \leq \theta\}$.  Then $\GM_{w \in A}(|Q(w)|) \geq 9/18^{\pi/\theta}$.
\end{theorem}
Our proof will require one standard fact from the theory of ``Mahler measures'':
\begin{fact}                                        \label{fact:mahler}
    Let $Q$ be a complex polynomial and let $\calO$ be a circle in the complex plane with center~$c$.  Then $\GM_{w \in \calO}(|Q(w)|) \geq |Q(c)|$.
\end{fact}
\begin{proof}
    By a linear transformation we may assume $\calO$ is the unit circle $\bdry D_1(0)$.  Express $Q(w) = a_0 \prod_{i}(w-\alpha_i)$, where the $\alpha_i$'s are the roots of~$Q$.  Then $\GM_{w \in \calO}(|Q(w)|)$ --- known as $Q$'s \emph{Mahler measure}, see e.g.~\cite{Smyth08} --- is exactly equal to $|a_0| \prod_{i \in I} |\alpha_i|$, where $I = \{i : |\alpha_i| \geq 1\}$.  (Since $\GM_{w \in \calO}(|\cdot|)$ is multiplicative, this statement follows immediately from the elementary fact that $\GM_{w \in \calO}(|w - \alpha|) = \max\{|\alpha|, 1\}$.) But clearly we have  $|a_0| \prod_{i \in I} |\alpha_i| \geq |a_0| \prod_i |\alpha_i| = |Q(0)|$.
\end{proof}

We can now establish Theorem~\ref{thm:BE-like}:
\begin{proof}[Proof of Theorem~\ref{thm:BE-like}]
    Using the bounds on $Q$'s coefficients we have:
    \begin{align}
        |Q(w)|   &\leq 1 + |w| + |w|^2 + \cdots = \frac{1}{1-|w|} \text{ for } w \in D_1(0); \label{eqn:bdry-bound0} \\
        |Q(1/3)| &\geq 1 - |1/3| - |1/3|^2 - \cdots = 1/2.
    \end{align}
    Let us apply Fact~\ref{fact:mahler} with $\calO = \bdry D_r(1/3) \supset A$, writing $A'$ for the complementary arc to~$A$ in~$\calO$.   We get
    \begin{equation}    \label{eqn:almost-done0}
       1/2 \leq \GM_{w \in \calO}(|Q(w)|) = \GM_{w \in A}(|Q(w)|)^{\theta/\pi} \cdot \GM_{w \in A'}(|Q(w)|)^{1-\theta/\pi}.
    \end{equation}
    And by~\eqref{eqn:bdry-bound0} we have
    \begin{equation} \label{eqn:combo1}
        \mathop{\GM}_{w \in A'}(|Q(w)|) \leq \mathop{\GM}_{w \in A'}(\tfrac{1}{1-|w|}) \leq \mathop{\GM}_{w \in \calO}(\tfrac{1}{1-|w|}) \leq \mathop{\GM}_{w \in \bdry D_{2/3}(1/3)}(\tfrac{1}{1-|w|}),
    \end{equation}
    where the second inequality is because the points $w \in A$ only have larger $\frac{1}{1-|w|}$ than the points in~$A'$, and the third inequality is because increasing the radius of $\calO$ from $r$ to $2/3$  only increases the value of $\frac{1}{1-|w|}$ for points on~$\calO$.  But now for $-\pi < t \leq \pi$, the point $w = 1/3 + (2/3)e^{it} \in D_{2/3}(1/3)$ has $|w|^2 = 1 - \frac49(1-\cos t)$ and hence
    \[
        \frac{1}{1-|w|} = \frac{1}{1-\sqrt{1 - \frac49(1-\cos t)}} \leq \frac{9}{2(1-\cos t)}.
    \]
    Thus
    \begin{equation}    \label{eqn:combo2}
        \mathop{\GM}_{w \in \bdry D_{2/3}(1/3)}(\tfrac{1}{1-|w|}) \leq \exp\parens*{\tfrac{1}{2\pi}\int_{-\pi}^{\pi}\ln\parens*{\tfrac{9}{2(1-\cos t)}}\,dt} = \frac92 \exp\parens*{-\tfrac{1}{2\pi}\int_{-\pi}^{\pi} \ln (1-\cos t)\,dt} = 9,
    \end{equation}
    the last integral being known.  (One can get a much easier integral, with a slightly worse constant, by lower-bounding $1-\cos t \geq (2/\pi^2)t^2$.)    Combining~\eqref{eqn:almost-done0}, \eqref{eqn:combo1}, \eqref{eqn:combo2} yields the theorem.
\end{proof}

\ignore{

START IGNORE

\subsection{XXXOlder version of the above; can probably junk it now XXX}

\begin{theorem} \label{thm:BE-like2}
    Let $Q(w)$ be a polynomial with constant coefficient~$1$ and all other coefficients in~$D_1(0)$.  Fix any $0 < \theta \leq \pi$, $0 \leq R \leq 2/3$, and let $A$ be the arc $\{1/3 + Re^{i t} : -\theta < t \leq \theta\}$.  Then the geometric mean of $|Q|$ on~$A$ satisfies $\GM_A(|Q|) \geq 9/18^{\pi/\theta}$.
\end{theorem}
\begin{proof}
    Using the bounds on $Q$'s coefficients we have:
    \begin{align}
        |Q(w)|   &\leq 1 + |w| + |w|^2 + \cdots = \frac{1}{1-|w|} \text{ for } w \in D_1(0); \label{eqn:bdry-bound} \\
        |Q(1/3)| &\geq 1 - |1/3| - |1/3|^2 - \cdots = 1/2. \label{eqn:center-bound}
    \end{align}
    We will also use the following elementary fact, which follows from the theory of Mahler measures:
    \begin{fact}                                        \label{fact:mahler}
        Let $Q$ be a complex polynomial and let $\calO$ be a complex circle with center~$c$.  Then the geometric mean of $|Q|$ around $\calO$ satisfies $\GM_{\calO}(|Q|) \geq |Q(c)|$.
    \end{fact}
    \begin{proof}
        By a linear transformation we may assume $\calO$ is the unit circle $\bdry D_1(0)$.  Express $Q(w) = a_0 \prod_{i=1}(w-\alpha_i)$, where the $\alpha_i$'s are the roots of~$Q$.  Then $\GM_{\calO}(|Q|)$ --- known as $Q$'s \emph{Mahler measure}, see~XXX, e.g.\ --- is exactly equal to $|a_0| \prod_{i \in I} |\alpha_i|$, where $I = \{i : |\alpha_i| \geq 1\}$.  (Since $\GM_{\calO}(|\cdot|)$ is multiplicative, this statement follows immediately from the elementary fact that $\GM_{\calO}(|w - \alpha|) = \max\{|\alpha|, 1\}$.) But evidently $|a_0| \prod_{i \in I} |\alpha_i| \geq |a_0| \prod_i |\alpha_i| = |Q(0)|$.
    \end{proof}
    Let us apply this fact with $\calO = \bdry D_R(1/3) \supset A$, writing $A'$ for the complementary arc to~$A$ in~$\calO$.   We get
    \begin{equation}    \label{eqn:almost-done1}
       1/2 \leq \GM_{\calO}(|Q|) = \GM_{A}(|Q|)^{\theta/\pi} \cdot \GM_{A'}(|Q|)^{1-\theta/\pi}.
    \end{equation}
    And by~\eqref{eqn:bdry-bound} we have
    \[
        \mathop{\GM}_{w \in A'}(|Q(w)|) \leq \mathop{\GM}_{w \in A'}(\tfrac{1}{1-|w|}) \leq \mathop{\GM}_{w \in \calO}(\tfrac{1}{1-|w|}) \leq \mathop{\GM}_{w \in \bdry D_{2/3}(1/3)}(\tfrac{1}{1-|w|}),
    \]
    where the second inequality is because the points $w \in A$ only have larger $\frac{1}{1-|w|}$ than the points in $A'$, and the third inequality is because increasing the radius $R$ only increases the value of $\frac{1}{1-|w|}$ for points on~$\calO$.  But now for $-\pi < t \leq \pi$, the point $w = 1/3 + (2/3)e^{it} \in D_{2/3}(1/3)$ has $|w| = \sqrt{1 - \frac89 \sin^2(t/2)} \leq 1-\frac49 \sin^2(t/2)$ and hence $\frac{1}{1-|w|} \leq \frac94 \csc^2 (t/2)$.  We conclude that
    \[
        \mathop{\GM}_{w \in \bdry D_{2/3}(1/3)}(\tfrac{1}{1-|w|}) \leq \exp\parens*{\tfrac{1}{2\pi}\int_{-\pi}^{\pi}\ln(\tfrac94 \csc^2 (t/2))\,dt} = \frac94 \exp\parens*{-\tfrac{1}{\pi}\int_{-\pi/2}^{\pi/2} \ln \sin^2 u\,du} = 9,
    \]
    the last integral being well known.  Putting these conclusions back into~\eqref{eqn:almost-done1} yields the lemma.
\end{proof}

\myfig{.5}{figure1c.png}{}{fig:1}

The following lemma is completely elementary:
\begin{lemma}                                       \label{lem:trig}
    Let $0 < r < 1/3$, and let  $0 < \theta \leq (3/2)r$.  Define $\Gamma = \frac{2/3 - r}{r} \geq 1$ and note that $\theta \leq (3/2)r \implies \sin\theta < 1/\Gamma$; this inequality is easily seen to be the precise condition needed so that the ray $\{1/3 + R e^{i\theta} : R > 0\}$ intersects $D_{r}(1-r)$ twice.  Let $w^*$ be the further of the two intersection points.  Then $|w^*|^2 \geq 1- \frac{8\theta^2}{9r}$.
\end{lemma}
\begin{proof}
    Let $T$ be the triangle with corners $1/3$, $1-r$, and let $\alpha$ denote the angle of~$T$ at $w^*$. (See Figure~\ref{fig:1}.) By the Law of Sines applied to~$T$,
    \[
        \frac{\sin\alpha}{2/3 - r} = \frac{\sin \theta}{r} \implies \sin \alpha = \Gamma \sin \theta \implies \cos \alpha = \sqrt{1-\Gamma^2 \sin^2 \theta}.
    \]
    Furthermore, using the fact that $w_\theta \in D_{r}(1-r)$, some basic trigonometry yields
    \begin{align}
        1 - |w^*|^2 = 2r(1-r)(1-\cos(\theta + \alpha)) &\leq 2r(1 - \cos \theta \cos \alpha + \sin \theta \sin \alpha) \nonumber\\
        &= 2r\parens*{1 - \sqrt{1-\sin^2\theta} \sqrt{1-\Gamma^2 \sin^2 \theta} + \Gamma \sin^2 \theta}. \label{eqn:AM-GM-sux}
    \end{align}
    If the AM-GM inequality were somehow reversed, we would have $\sqrt{1-\sin^2\theta} \sqrt{1-\Gamma^2 \sin^2 \theta} \geq 1 - (\frac12 + \frac12 \Gamma^2) \sin^2 \theta$, which would give us a simple upper bound for~\eqref{eqn:AM-GM-sux}.  Unfortunately, the AM-GM inequality is not reversed, but one can  show in this case\onote{Don't ask me how, but I guarantee it\dots}  that the deficit is bounded by $\frac12 \Gamma^4 s^4$. That is,
    \[
        \sqrt{1-\sin^2\theta} \sqrt{1-\Gamma^2 \sin^2 \theta} \geq 1 - \parens*{\frac12 + \frac12 \Gamma^2} \sin^2 \theta - \frac12 \Gamma^4 s^4,
    \]
    and hence
    \[
        \eqref{eqn:AM-GM-sux} \leq r(1+\Gamma)^2 \sin^2 \theta + r \Gamma^4 s^4 = \frac{4}{9r}\sin^2 \theta + r \parens*{\frac{2/3-r}{r}}^4 \sin^4 \theta \leq \frac{4\theta^2}{9r}\parens*{1+ \frac{4\theta^2}{9r^2}} \leq \frac{8\theta^2}{9r},
    \]
    the last inequality using $\theta \leq (3/2)r$. The proof is complete.
\end{proof}

We can now prove Theorem~\ref{thm:Lit-chill-lower}.  Assume $1/d^{1/2} \leq r \leq .1$ and let $P(w) = w^dQ(w)$, where $Q(w)$ is a polynomial with constant coefficient~$1$ and all other coefficients in~$D_1(0)$. Our task is to show
\begin{equation}    \label{eqn:prove-me-1}
    \max_{w \in D_r(1-r)} \abs*{P(w)} \geq (4040)^{-(d/r)^{1/3}}.
\end{equation}
Select $\theta = \frac{3}{2} (r/d)^{1/3}$; this satisfies $\theta \leq (3/2)r$ because $r \geq 1/d^{1/2}$. Set $w^*$ as in Lemma~\ref{lem:trig} and let $A$ be the arc as in Theorem~\ref{thm:BE-like} with endpoints $w^*, \overline{w^*}$.  Note that $A \subseteq D_r(1-r)$ (using $r \leq .1 < 2/3$), and that $|w| \geq |w^*|$ for all $w \in A$.  From the latter fact and Lemma~\ref{lem:trig} we get
\begin{equation}    \label{eqn:combine1}
    w \in A \implies |w|^d \geq \parens*{1-\frac{8\theta^2}{9r}}^{d/2} \geq \exp(-\theta^2/r)^{d/2} = \exp(-d\theta^2/r) = \exp(-(9/4)(d/r)^{1/3});
\end{equation}
here, the second inequality used $r \leq .1$.  At the same time, Theorem~\ref{thm:BE-like} tells us
\begin{equation}    \label{eqn:combine2}
    \GM_A(|Q|) \geq 9/18^{\pi/\theta} \geq \exp(-\ln(18)\pi/\theta) = \exp(-(2/3)\ln(18)\pi (d/r)^{1/3}).
\end{equation}
Combining~\eqref{eqn:combine1},  \eqref{eqn:combine2}, and the fact that $9/4 + (2/3)\ln(18)\pi \leq \ln(4040)$, yields
\[
    \GM_A(|P|) \geq (4040)^{-(d/r)^{1/3}},
\]
which is stronger than~\eqref{eqn:prove-me-1} because $A \subseteq D_r(1-r)$.

END IGNORE

}

\section{Proof of Theorem~\ref{thm:Lit-upper}}
The key ingredient is the following theorem from \cite{BEK99}.  (Recall that a \emph{Littlewood polynomial} has all nonzero coefficients either $-1$ or 1.)
\begin{theorem}
[\cite{BEK99}, Theorem~3.3] \label{thm:BEK}
For all $k \geq 2$ there is a nonzero Littlewood polynomial $Q_k$ of degree at most~$k$ satisfying $|Q_k(t)| \leq \exp(-c_0 \sqrt{k})$ for all real $0 \leq t \leq 1$.  Here $c_0 > 0$ is a universal constant.
\end{theorem}
By a simple use of the Hadamard Three-Circle Theorem and Maximum Modulus Principle, Borwein and Erd\'elyi proved in \cite{BE97} that the polynomials in Theorem~\ref{thm:BEK} establish tightness of their Theorem~\ref{thm:be} (up to the constant~$C_1$).  We quote a result that appears within their proof:
\begin{theorem}                     [\cite{BE97}, in the first proof of Theorem 3.3 in the ``special case'', p. 11]                \label{thm:BE-ellipse}
 There are universal constants $c_1, c_2, c_3 > 0$ such that the following holds: For all $0 < a \leq c_1$ there exists an integer $2 \leq k \leq c_2/a^2$ such that $\max_{w \in D_{6a}(1)} |Q_k(w)| \leq \exp(-c_3/a)$, where $Q_k$ is the nonzero Littlewood polynomial from Theorem~\ref{thm:BEK}.
\end{theorem}
\begin{remark}
    Actually, Borwein and Erd\'elyi proved this with an elliptical disk $\calE_a$ in place of $D_{6a}(1)$, where $\calE_a$ has foci at $1-8a$ and~$1$ and major axis $[1-14a, 1+6a]$. It is easy to see that $D_{6a}(1) \subset \calE_a \subset D_{14a}(1)$, so we wrote $D_{6a}(1)$ in Theorem~\ref{thm:BE-ellipse} for simplicity and because it loses almost nothing.
\end{remark}

\newcommand{\Exp}{\mathrm{Exp}}
We can now prove Theorem~\ref{thm:Lit-upper}.  We state here a slightly more precise version:
\begin{theorem}                                     \label{thm:Lit-upper-new}
    Using the notation $\delta = 1-\rho$, and the notation $\Exp(t) = \exp(c \cdot t)$ for an unspecified universal constant $c > 0$, we have
    \[
        \kappa_{\Lit}(\rho,n) \leq \begin{cases}
                                                        \Exp(-(\delta n)^{1/3}) & \text{in \textnormal{Case~I}:\phantom{II} $C(\log^3 n)/n \leq \delta \leq 1/2$,} \\
                                                        \Exp(-(n/\rho)^{1/3}) & \text{in \textnormal{Case~II}:\phantom{I} $C/n^{1/2} \leq \rho \leq 1/2$,}
                                                  \end{cases}
    \]
    provided $n \geq n_0$.  Here $n_0, C \geq 1$ are universal constants.
\end{theorem}
\begin{proof}[Proof of Theorem~\ref{thm:Lit-upper}.]
    With $C \geq 1$ to be specified later, select
    \[
        a = \begin{cases}
                    C_1/(\delta n)^{1/3} & \text{in Case~I:\phantom{II} $C(\log^3 n)/n < \delta \leq 1/2$,} \\
                    C_1(\rho/n)^{1/3} & \text{in Case~II:\phantom{I} $1/n^{1/2} < \rho < 1/2$,}
                \end{cases}
    \]
    where $C_1 \geq 1$ is a universal constant to be specified later. Assuming $n_0 = n_0(C_1)$ is sufficiently large we get that $a \leq c_1$, where $c_1$ is as in Theorem~\ref{thm:BE-ellipse}. Applying that theorem, we obtain
    \begin{equation}    \label{eqn:A}
        \max_{w \in A} |Q_k(w)| \leq \Exp(-1/a), \qquad \text{where }
        A \coloneqq D_{6a}(1), \quad k \leq c_2/a^2 < n/2.
    \end{equation}
    Here the inequality $c_2/a^2 < n/2$ holds in Case~I by assuming $n_0 = n_0(C_1,c_2)$ large enough, and in Case~II by taking $C_1 = C_1(c_2)$ large enough.  Now define
    \[
        P(w) = w^{\lfloor n/2 \rfloor} \cdot Q_k(w), \quad \text{a nonzero Littlewood polynomial of degree less than~$n$.}
    \]
    We wish to bound
    \[
        \max_{w \in R} |P(w)|, \qquad R \coloneqq D_{\rho}(\delta)
    \]
    by the expression in the theorem statement.  For the points $w \in R \cap A$, we are done by~\eqref{eqn:A} (and the fact that $|w^{\lfloor n/2 \rfloor}| \leq 1$).  For the points in $w \in R \setminus A$, we claim that
    \begin{equation}    \label{eqn:RminusA}
        |w|^2 \leq 1 - 36 \tfrac{\delta}{\rho}a^2  \leq \exp(-36 \tfrac{\delta}{\rho}a^2) \quad \forall w \in R \setminus A.
    \end{equation}
    Assuming~\eqref{eqn:RminusA}, we get
    \[
        \max_{w \in R \setminus A} |P(w)| \leq \max_{w \in R \setminus A} |w|^{\lfloor n/2 \rfloor} \cdot \max_{w \in R \setminus A} |Q_k(w)| \leq \exp(-18 \tfrac{\delta}{\rho}a^2)^{\lfloor n/2 \rfloor} \cdot (n/2+1) \leq \Exp(-n \tfrac{\delta}{\rho}a^2) \cdot (n/2+1),
    \]
    where the factor $n/2+1$ is an upper bound on $|Q_k(w)|$ over all of $D_1(0)$ (recall that $Q_k$ is a Littlewood polynomial of degree less than $n/2$).  By inspection, this is sufficient to complete the proof in both Case~I and Case~II (in Case~I we need to assume $C$ large enough to absorb the factor of $(n/2+1)$).

    It remains to establish~\eqref{eqn:RminusA}.  For this we first note that $\rho > 3a$ in both Case~I and Case~II (Case~I is easier to check; for Case~II we need to use that $C = C(C_1)$ is sufficiently large).  This in particular means that $R \setminus A \neq \emptyset$.  Writing $w_0$ for either of the intersection points of $\bdry R$ and $\bdry A$, we have $\max_{w \in R \setminus A} |w| \leq |w_0|$.  Thus it suffices to upper-bound $|w_0|^2$.

    In the complex plane, consider the triangle formed by $\delta$, $1$, and $w_0$.  Note that $w_0$ has distance~$\rho$ from $\delta$ and distance $6a$ from~$1$.  Let $\theta$ denote the triangle's angle at~$\delta$.  By the Cosine Law, $(6a)^2 = \rho^2 + \rho^2 - 2\rho^2 \cos \theta$ and hence $\cos\theta  = 1-18a^2/\rho^2$. Now consider the triangle formed by $\delta$, $0$, and $w_0$.  Its angle at~$\delta$ is $\pi - \theta$ and the adjacent sides have length $\delta$, $\rho$.  Thus by the Cosine Law,
    \[
        |w_0|^2 = \delta^2 + \rho^2 - 2\delta\rho\cos(\pi - \theta) = \delta^2 + \rho^2 + 2\delta\rho\cos \theta = (\delta + \rho)^2 - 36\delta \rho a^2/\rho^2 = 1 - 36 \tfrac{\delta}{\rho} a^2,
    \]
    as needed for~\eqref{eqn:RminusA}.
\end{proof}

\ignore{
\newpage
\section{From strings to polynomials}
\ignore{
%
%
}

\subsection{Characterizing the sample complexity of mean-based algorithms}

In this subsection we define a quantity $\vareps_\delta(n)$ and show that it characterizes, up to polynomial factors, the information-theoretic sample complexity of mean-based trace reconstruction algorithms.

We start with some simple analysis of the mean trace vector $\mu_\delta(x)$.  For $x \in \zo^n$ and $j \in \{0,\dots,n-1\}$, it is easy to verify that the $j$-th coordinate of $\mu_\delta(x)$ is
\begin{equation} \label{eq:mudeltaj}
\mu_\delta(x)_j = \Pr_{\by \leftarrow \Del_\delta(x)}[\by_j = 1] = \sum_{k = j}^{n-1} x_k \Pr[k \rightsquigarrow j] = \sum_{k = j}^{n-1}  \rho \Pr[\text{Bin}(k,\rho)=j]x_j,
\end{equation}
where ``$\Pr[k \rightsquigarrow j]$'' denotes the probability that coordinate $x_k$ of $x$ appears as coordinate $\by_j$ of $\by \leftarrow \Del_\delta(x).$ By (\ref{eq:mudeltaj}) we may write the mean trace vector $\mu=\mu_\delta(x)$ as $\mu = \rho A x$ where $A=A(\delta)$ is the $n$-by-$n$ matrix given by
\begin{equation} \label{eq:A}
    A = (a_{jk})_{j,k=0}^{n-1}, \quad a_{jk} = \Pr[\text{Bin}(k,\rho)=j].
\end{equation}
Note that for $0 < \delta < 1$ the matrix $A$ is upper diagonal and non-singular.

We define a crucial quantity:

\begin{definition} \label{def:eps-infty}
We define the quantity $\vareps_\delta(n)$ (or $\vareps$ for short when $\delta,n$ are clear from context) to be
\begin{equation} \label{eqn:eps-infty}
\vareps_\delta(n) := \min_{x\neq x' \in \zo^n} \|\mu_\delta(x) - \mu_\delta(x')\|^2, \quad \text{or equivalently,} \quad
\vareps_\delta(n) := \rho ^2 \cdot \min_{\substack{u \in \{-1,0,1\}^n \\ u \neq 0^n}} \|A u\|^2,
\end{equation}
the minimum possible squared Euclidean distance between any two distinct mean trace vectors.
\end{definition}

Note that since $A$ is non-singular we have that $\vareps_\delta(n)>0.$

It is straightforward to obtain an upper bound on the optimal sample complexity of mean-based trace reconstruction algorithms in terms of $\vareps_\delta(n):$

\begin{lemma} \label{lem:upper}
For any $0 < \delta < 1$, there is an $m$-sample mean-based trace reconstruction algorithm at deletion rate $\delta$ that has sample complexity $m=O({\frac {n \log n}{\vareps_\delta(n)}})$.
\end{lemma}

\begin{proof}
Let $x$ be the unknown source string.  By a simple Chernoff bound and union bound, for $m=O(n \log n /\vareps_\delta(n))$ it is the case that with probability $0.999$ an outcome $\widehat{\mu}_{m,\delta}(x)$ drawn from $\widehat{\bmu}_{m,\delta}(x)$ will differ from $\mu_\delta(x)$ by at most $\sqrt{\vareps_\delta(n)/(3n)}$ in each coordinate, and hence will satisfy
$\|\mu_\delta(x) - \widehat{\mu}_{m,\delta}(x)\|^2 \leq \vareps_\delta(n)/3.$  Given this vector $\widehat{\mu}_{m,\delta}(x)$, an exhaustive enumeration of all $2^n$ possible mean traces $\mu_\delta(x')$ as $x'$ varies across all of $\zo^n$ will yield a unique $x' \in \zo^n$ for which $\|\mu_\delta(x') - \widehat{\mu}_{m,\delta}(x)\|^2 \leq \vareps_\delta(n)/3,$ and this $x'$ must be the source string $x$.
\end{proof}

Next we establish a lower bound on the sample complexity of mean-based algorithms in terms of $\vareps_\delta(n)$:

\begin{lemma} \label{lem:lower}
Any $m$-sample mean-based trace reconstruction algorithm must have $m=\Omega(\sqrt{\frac {1}{n \cdot \vareps_\delta(n)}}).$
\end{lemma}
\begin{proof}
Fix $x \neq x' \in \zo^n$ to be such that $\vareps_\delta(n) = \|\mu_\delta(x) - \mu_\delta(x')\|^2$.  For each $j \in \{0,\dots,n-1\}$ it must be the case that $|\mu_\delta(x)_j - \mu_\delta(x')_j| \leq \sqrt{\vareps_\delta(n)/n}.$  Since the random variable $\widehat{\bmu}_{m,\delta}(x)_j$ is distributed according to ${\frac 1 m} \cdot \Bin(m,\mu_\delta(x)_j)$, and the variation distance between two Binomial distributions $\Bin(m,p)$ and $\Bin(m,q)$ is at most $m(p-q)$, we have $\dtv{\widehat{\bmu}_{m,\delta}(x)_j}{\widehat{\bmu}_{m,\delta}(x')_j} \leq m\sqrt{\vareps_\delta(n)/n},$ and hence by the triangle inequality $\dtv{\widehat{\bmu}_{m,\delta}(x)}{\widehat{\bmu}_{m,\delta}(x')} \leq m\sqrt{n \cdot \vareps_\delta(n)}.$  But $\dtv{\widehat{\bmu}_{m,\delta}(x)}{\widehat{\bmu}_{m,\delta}(x')}$ must certainly be at least (say) $1/2$, since otherwise it is impossible for an $m$-sample mean-based trace reconstruction algorithm to output the correct answer with probability at least $99/100$ both when run on the source string $x$ and when run on $x'$.  The lemma follows.
\end{proof}

Lemmas~\ref{lem:upper} and~\ref{lem:lower} show that up to $\poly(n)$ factors (which are negligible given the bounds that we will ultimately establish on $1/\vareps_\delta(n)$), the information-theoretic sample complexity of mean-based algorithms for trace reconstruction algorithms is $1/(\vareps_\delta(n))^c$ for some $\half \leq c \leq 1.$\rnote{From Ryan's old notes:  ``It's not extremely clear to me, but we kind of feel that Mitz et al.'s Theorem~3.1 is suggesting
\[
    \exp(-\wt{O}(\sqrt{n}))\ \mathop{\leq}^{?}\ \vareps_\delta(n)\ \mathop{\leq}^{?}\ n^{-\wt{\Omega}(\log n)}.\text{''}
\]
Verify this and add some discussion of it here?}

 With these lemmas in hand several tasks remain for us.  First, of course we would like to give a concrete upper bound on $\vareps_\delta(n)$ and thereby via Lemma~\ref{lem:lower} obtain a lower bound on the sample complexity of mean-based trace reconstruction.  In Section~\ref{sec:string-to-poly} we give a characterization of $\vareps_\delta(n)$ in terms of polynomials over $\C$ which will later lead us to an upper bound on it.  Second, while Lemma \ref{lem:upper} gives a sample complexity upper bound in terms of $\vareps_\delta(n)$, the algorithm described in the proof of that lemma is highly inefficient because of enumerating over all candidate strings $x' \in \zo^n$;  we would like to have a computationally efficient (as well as sample-efficient) algorithm.  In Section~\ref{sec:beyond-eps} we define a relaxation $\vareps^{\chill}_\delta(n)$ that is closely related to $\vareps_\delta(n)$ and  show that there is an algorithm with sample complexity \emph{and running time} polynomial in $n$, $1/\vareps^{\chill}_\delta(n)$ and $1/\rho.$

Looking ahead, once the results of Sections~\ref{sec:string-to-poly} and~\ref{sec:beyond-eps} are in hand, the main technical contributions of our paper are in Sections~\ref{sec:A} and~\ref{sec:B}.  In Section~\ref{sec:A} we give an upper bound on $\vareps_\delta(n)$ and thereby establish Theorem~\ref{thm:main-positive}. In Section~\ref{sec:B} we give a lower bound on $\vareps^{\chill}_\delta(n)$ and thereby establish Theorem~\ref{thm:main-negative}.  Happily, the lower bound we obtain on the relaxation $\vareps^{\chill}_\delta(n)$ is within polynomial factors of our upper bound on $\vareps_\delta(n)$, allowing us to obtain essentially matching upper and lower bounds overall.

\subsection{A reformulation of $\vareps_\delta(n)$ in terms of polynomials over $\C$} \label{sec:string-to-poly}

Recalling (\ref{eq:A}), we would like to understand the vector $Au$ as $u$ ranges over $\{-1,0,1\}^n \setminus \{0^n\}$.  We take a generating function approach.  Writing $Z$ for the column vector $(1, z, z^2, \dots, z^{n-1})$, we consider the ordinary generating function $G(Au;Z),$
\[
G(Au;Z) := (Au)^\top Z = u^\top A^\top Z.
\]

Note that the $k$th row of $A^\top$ is just the pmf of the random variable $\Bin(k,\rho)$, so the $j$-th entry in row $k$ is
${k \choose j} \rho^j \delta^{k-j}.$  It follows from the binomial theorem that the $k$-th entry of $A^\top Z$ is $(\delta + \rho z)^k$, and hence the generating function  is
\begin{equation} \label{eq:gf}
G(Au,Z)= \sum_{j = 0}^{n-1} u_j (\delta + \rho z)^j =:  p_u(\delta + \rho z),
\end{equation}
where
\[
    p_u(w) := \sum_{j=0}^{n-1} u_j w^j
\]
is a nonzero ``Littlewood polynomial,'' i.e. a  polynomial all of whose coefficients belong to $\{-1,0,1\}.$

The polynomial $p_u(w)$ is where we make our segue to working with complex numbers.  We recall the following :

\begin{fact} \label{fact:length-modulus}
For any real vector $v=(v_0,\dots,v_{n-1})$, the squared Euclidean length of $v$ is equal to the average of the squared modulus of $G(v;Z)$ averaged around the unit circle of the complex plane.
\end{fact}
\begin{proof}
\rnote{can/should we cite something here instead of writing this proof?  I guess it's short...}The average squared modulus of $G(v;Z)$ is
\begin{align*}
\Avg_{|z|=1} \left| \sum_{j=0}^{n-1} v_j z^j \right|^2
&= {\frac 1 {2\pi}} \int_{0}^{2 \pi}  \left| \sum_{j=0}^{n-1} v_j e^{ijt} \right|^2 dt
= {\frac 1 {2\pi}} \int_{0}^{2 \pi}  \left| \sum_{j=0}^{n-1} v_j \cos jt + i \sum_{j=0}^{n-1} v_j \sin jt \right|^2 dt\\
&= {\frac 1 {2\pi}} \int_{0}^{2 \pi}  \left(\sum_{j=0}^{n-1} v_j \cos jt\right)^2 + \left(\sum_{j=0}^{n-1} v_j \sin jt \right)^2 dt\\
&= {\frac 1 {2\pi}} \int_{0}^{2 \pi}  \sum_{j=0}^{n-1} v_j^2 (\cos^2 jt + \sin^2 jt)
+ \sum_{j_1 \neq j_2} v_{j_1} v_{j_2} ( \cos j_1 t \cos j_2 t + \sin j_1 t \sin j_2 t)
 dt\\
 &= \sum_{j=0}^{n-1} v_j^2 = \|v\|^2.
\end{align*}
\end{proof}

As a direct consequence, from (\ref{def:eps-infty}) and (\ref{eq:gf}) we have
\begin{align}
\vareps_\delta(n) &= \rho ^2 \cdot \min_{\substack{u \in \{-1,0,1\}^n \\ u \neq 0^n}}
\Avg_{|z|=1} |G(Au;Z)|^2 \nonumber \\
&= \rho^2 \cdot \min\left\{ \Avg_{z \in \partial D_\rho(\delta)} |p(z)|^2  \ : \ p \text{~is a nonzero Littlewood polynomial of degree $< n$}\right\}. \label{eq:avg}
\end{align}

The following lemma shows that up to polynomial factors the average above can be replaced by the maximum:

\begin{lemma} \label{lem:avg-max}
For any nonzero Littlewood polynomial $p$ of degree at most $n$ and any $0 < \rho < 1$,
\[
{\frac {\sqrt{{\max_{z \in \partial D_{\rho}(\delta)} |p(z)|^2}}} {150\delta n^2}} \leq
\Avg_{z \in \partial D_{\rho}(\delta)} |p(z)|^2 \leq
\max_{z \in \partial D_{\rho}(\delta)} |p(z)|^2.
\]
\end{lemma}
\begin{proof}
The second inequality is obvious.  For the first, let $\gamma$ denote $\max_{z \in \partial D_{\rho}(\delta)} |p(z)|$ so
$\gamma^2 = \max_{z \in \partial D_{\rho}(\delta)} |p(z)|^2.$  Let $z^\ast \in \partial D_\rho(\delta)$ be such that $|p(z^\ast)|=\gamma.$  Since $p$ is a degree-$n$ Littlewood polynomial its derivative has magnitude at most $n^2$ at every point on $\partial D_\rho(\delta)$, and consequently $|p(z)| \geq \gamma/2$ for every point within angular distance $1/(10\delta n^2 \gamma)$ of $z^\ast$ on $\partial D_\rho(\delta).$  It follows that
\[
\Avg_{z \in \partial D_\rho(\delta)} |p(z)|^2 \geq {\frac 1 {10 \delta n^2 \gamma}} \cdot {\frac 1 \pi} \cdot {\frac {\gamma^2} 4},
\]
giving the first inequality.
\end{proof}

By the Maximum Modulus principle, any polynomial $p$ has
\begin{equation} \label{eq:maxmod}
\max_{z \in D_\rho(\delta)}|p(z)| =
\max_{z \in \partial D_\rho(\delta)}|p(z)|.
\end{equation}
Motivated by Lemma \ref{lem:avg-max} and (\ref{eq:maxmod}), we define
\[
    \vareps^\Lit_\delta(n) := \min \left\{\max_{z \in D_\rho(\delta)} |p(z)|  \ : \ p \text{~is a nonzero Littlewood polynomial of degree $< n$}\right\}.
\]

In Section \ref{sec:A} we prove the upper bound
\begin{equation} \label{eq:ub}
\vareps^\Lit_\delta(n)\leq \exp(-\Omega(n^{1/3}/\rho^{1/3})).
\end{equation}
Combining Lemma \ref{lem:lower},  (\ref{eq:avg}), Lemma \ref{lem:avg-max}, (\ref{eq:maxmod}) and (\ref{eq:ub}), we obtain Theorem~\ref{thm:main-negative}.

\subsection{A relaxation that leads to an efficient algorithm}\label{sec:beyond-eps}

We define a nonzero polynomial $p(x)=\sum_{j=0}^{n-1} a_j x^j$ with real coefficients to be a \emph{chill} polynomial if
\begin{enumerate}
\item the smallest $j$ such that $a_j \neq 0$ has $|a_j|=1$, and
\item every $|a_j| \leq 1.$
\end{enumerate}

Note that every nonzero Littlewood polynomial is a chill polynomial.

For our algorithmic result we consider the following relaxation $\vareps^{\chill}_\delta(n) \leq \vareps^\Lit_\delta(n)$:
\[
\vareps^{\chill}_\delta(n) := \min\left\{\max_{z \in D_\rho(\delta)} |p(z)|  \ :\ p \text{~is a chill polynomial of degree $< n$}\right\}.
\]

We show in Section~\ref{sec:B} that $\vareps^{\chill}_\delta(n)$ is not too small.  Together with the following theorem, which should be compared with Lemma~\ref{lem:upper}, this gives our main positive result, Theorem~\ref{thm:main-positive}:

\begin{theorem} \label{thm:alg}
For any $0 < \delta < 1$, there is an $m$-sample mean-based trace reconstruction algorithm at deletion rate $\delta$ that has sample complexity $m = \poly(n,1/\vareps^\chill_\delta(n),1/\rho)$ and runs in time $\poly(m,n).$
\end{theorem}

The high-level idea of Theorem~\ref{thm:alg} is to solve a sequence of $n$ convex minimization problems, each of which lets us discover one bit of the source string $x$.
We establish some lemmas that will be used in the proof of Theorem~\ref{thm:alg}.  The first simple lemma shows that it is easy for a mean-based algorithm to determine $\|x\|_1$, the number of ones in the source string $x$:

\begin{lemma} \label{lem:numones}
For $m = O({\frac {n^2 \log n} {\rho^2}})$, with probability at least $1-1/n^3$ over a draw of $\widehat{\mu}_{m,\delta}(x)$ from $\widehat{\bmu}_{m,\delta}(x)$ the value
$\left \lfloor \|\widehat{\mu}_{m,\delta}(x)\|_1/\rho \right \rceil$ is equal to $\|x\|_1.$
\end{lemma}
\begin{proof} Since every one-coordinate of $x$ is independently retained in $\by \leftarrow \Del_\delta(x)$ with probability $\rho$, by linearity of expectation we have $\|\mu_{\delta}(x)\|_1 = \rho\|x\|_1.$  For each coordinate $j \in \{0,\dots,n-1\}$ a simple Chernoff bound gives that $|\widehat{\mu}_{m,\delta}(x)_j - \mu_\delta(x)_j| \leq {\frac \rho {10n}}$ except with failure probability at most $1/n^4,$ so with overall probability at least $1-1/n^3$ we have $| \| \widehat{\mu}_{m,\delta}(x)\|_1 - \|\mu_{\delta}(x)\|_1| \leq \rho/10$, which gives the lemma.
\end{proof}
Via Lemma \ref{lem:numones}, we can subsequently assume that a mean-based trace reconstruction algorithm is provided with the value $\|x\|_1.$

Our next lemma establishes an identity which is useful because it equates two polynomials, one of which ((\ref{eq:poly-equal2})) is defined in terms of the unknown bits $x_k$ of the source string and the other of which ((\ref{eq:poly-equal1})) is defined in terms of the mean trace vector, which can be easily estimated by a mean-based algorithm.  (Lemma~\ref{lem:estimate-poly} gives the details of performing such an estimation.)

\begin{lemma} \label{lem:poly}
For any source string $x \in \zo^n$ and any $0<\delta<1$, we have
\begin{align}
q_{x,\delta}(z) &:= n - \|\mu_\delta(x)\|_1 + \sum_{j=0}^{n-1} \mu_{\delta}(x)_j z^j  \label{eq:poly-equal1}\\
&= \rho \cdot\left( n - \|x\|_1 + \sum_{k=0}^{n-1} x_k (\delta + \rho z)^k \right).  \label{eq:poly-equal2}
\end{align}
\end{lemma}
\begin{proof}
We consider the quantity
$
\sum_{j=0}^{n-1} \E_{\by \leftarrow \Del_\delta(x)} \left[z^{\by_j j}\right].
$
It is easy to see that this is equal to $\sum_{j=0}^{n-1} ((1-\mu_{\delta}(x)_j) + \mu_\delta(x)_j z^j),$ giving (\ref{eq:poly-equal1}).  On the other hand, recalling the notation ``$\Pr[k \rightsquigarrow j]$'' that denotes the probability that coordinate $x_k$ of $x$ appears as coordinate $\by_j$ of $\by \leftarrow \Del_\delta(x),$ we have that for $j \leq k \leq n-1$,
\[
\Pr[k \rightsquigarrow j] =  {{k}\choose{j}} \rho^{j+1} \delta^{k-j},
\]
and hence
\[
\E_{\by \leftarrow \Del_\delta(x)}[z^{\by_j j}] = \sum_{k=j}^{n-1} \Pr[k \rightsquigarrow j] ( \mathbf{1}[x_k=0] + \mathbf{1}[x_k=1] \cdot  z^j)
=\sum_{k=j}^{n-1} {{k}\choose{j}} \rho^{j+1} \delta^{k-j} z^{j x_k}.
\]
Consequently we also have
\begin{align*}
 \sum_{j=0}^{n-1} \E_{\by \leftarrow \Del_\delta(x)} \left[z^{\by_j j}\right] &=
 \sum_{j=0}^{n-1} \sum_{k=j}^{n-1} {{k}\choose{j}} \rho^{j+1} \delta^{k-j} z^{j x_k}\\
 &= \sum_{k=0}^{n-1} \sum_{j=0}^{k} {{k}\choose{j}} \rho^{j+1} \delta^{k-j} z^{j x_k}\\
 &= \sum_{k=0}^{n-1} \rho \delta^k \sum_{j=0}^k {k \choose j} \left({\frac \rho \delta} \cdot z^{x_k}\right)^j\\
 &= \sum_{k=0}^{n-1} \rho \delta^k \left(1 + {\frac \rho \delta} \cdot z^{x_k}\right)^k \tag{by the binomial theorem}\\
 &= \rho \sum_{k=0}^{n-1} \left(\delta + \rho z^{x_k}\right)^k.\end{align*}
Observing that when $x_k=0$ we have $\left(\delta + \rho z^{x_k}\right)^k=1$ and when $x_k=1$ we have $\left(\delta + \rho z^{x_k}\right)^k= (\delta + \rho z)^k = x_k (\delta + \rho z)^k$, we get
\[
\sum_{j=0}^{n-1} \E_{\by \leftarrow \Del_\delta(x)} \left[z^{\by_j j}\right] =
\rho \cdot \left( n - \|x\|_1 + \sum_{k=0}^{n-1} x_k (\delta + \rho z)^k \right),
\]
establishing the lemma.
\end{proof}

\begin{lemma} \label{lem:estimate-poly}
Fix $x \in \zo^n$, $0 < \delta < 1$, and $z \in \C, |z| \leq 1.$  Given a draw of $\widehat{\mu}_{m,\delta}(x) \leftarrow \widehat{\bmu}_{m,\delta}(x)$ where $m := m(\kappa,\tau) =  O({\frac {(\rho n)^2 \log(n/\tau)}{\kappa^2}})$, with probability at least $1-\tau$ the value
\begin{equation} \label{eq:qhat}
\widehat{q}_{x,\delta}(z) := n - \|\widehat{\mu}_{m,\delta}(x)\|_1 + \sum_{j=0}^{n-1} \widehat{\mu}_{m,\delta}(x)_j z^j
\end{equation}
satisfies
$
|\widehat{q}_{x,\delta}(z) - q_{x,\delta}(z)| \leq \kappa.
$
\end{lemma}

\begin{proof}
Similar to the proof of Lemma~\ref{lem:numones}, for each $j \in \{0,\dots,n-1\}$ we have
\[
|\widehat{\mu}_{m,\delta}(x)_j z^j - \mu_\delta(x)_j z^j| = |z^j| \cdot
|\widehat{\mu}_{m,\delta}(x)_j - \mu_\delta(x)_j| \leq
|\widehat{\mu}_{m,\delta}(x)_j - \mu_\delta(x)_j| \leq
{\frac \kappa {2\rho n}}
\]
with probability at least $1-{\frac \tau n}$, where the last inequality is a simple Chernoff bound.  Hence by a union bound we have
\[
\left|
\|\widehat{\mu}_{m,\delta}(x)\|_1  - \|\mu_\delta(x)_j\|_1
\right| \leq {\frac \kappa {2\rho}}
\quad \text{and} \quad
\left|
\sum_{j=0}^{n-1} \widehat{\mu}_{m,\delta}(x)_j z^j -
\sum_{j=0}^{n-1} \mu_{\delta}(x)_j z^j
\right|
\leq {\frac \kappa {2\rho}}
\]
with probability at least $1-\tau$, and another union bound gives the lemma.
\end{proof}

Fix
\begin{equation} \label{eq:s}
s=100n^2/\vareps^\chill_\delta(n).
\end{equation}
Let $z_1,\dots,z_s$ be $s$ points equally spaced around the unit circle in the complex plane, and let $z'_i = \delta + \rho z_i$, $i=1,\dots,s$ be the corresponding set of points equally spaced around the circle $\partial D_{\rho}(\delta).$

\begin{lemma} \label{lem:empirical-min-close}
Let $m = m(\rho\vareps^\chill_\delta(n)/8,1/(n^2 s))$ be as defined in Lemma \ref{lem:estimate-poly}.
For any source string $x \in \zo^n$ and any $0 < \delta < 1$, with probability at least $1-1/n^2$ (over the draw of $\widehat{\mu}_{m,\delta}(x) \leftarrow \widehat{\bmu}_{m,\delta}(x)$ that determines $\widehat{q}_{x,\delta}(\cdot)$ via (\ref{eq:qhat})) we have that every $i \in [s]$ satisfies
\begin{equation} \label{eq:qhat-close-to-q}
|\widehat{q}_{x,\delta}(z_i) - q_{x,\delta}(z_i)| \leq \rho\vareps^{\chill}_\delta(n)/8.
\end{equation}
\ignore{
}
\end{lemma}

\begin{proof}
We say that a draw of $\widehat{\mu}_{m,\delta}(x) \leftarrow \widehat{\bmu}_{m,\delta}(x)$ that satisfies Lemma \ref{lem:estimate-poly} is \emph{good}, and we consider such a good draw.  Applying Lemma \ref{lem:estimate-poly} to each $z_i, i \in [s]$, via a union bound we get that with probability at least $1-1/n^2$ every $i \in [s]$ satisfies
(\ref{eq:qhat-close-to-q}) as required by the lemma.
\ignore{
}
\end{proof}

\medskip
\noindent \emph{Proof of Theorem~\ref{thm:alg}.}
The algorithm proceeds in $n$ stages $k=0,\dots,n-1$, where in stage $k$ it attempts to determine the bit $x_k$ of the source string $x$.  Suppose inductively that the algorithm has succeeded in stages $0,\dots,k-1$; we show below how the algorithm can correctly determine the bit $x_k$ with correctness probability $1-1/n^2$ in time $\poly(n,1/\vareps^\chill_\delta(n),1/\rho)$, which clearly establishes the theorem.

By assumption the algorithm has the correct values $x_0,\dots,x_{k-1} \in \zo$.  The idea of stage $k$ is that in it the algorithm guesses that $x_k=0$ and checks whether this guess is correct.  This is done as follows:

\begin{enumerate}

\item Draw an outcome $\widehat{\mu}_{m,\delta}(x) \in [0,1]^n$  from $\widehat{\bmu}_{m,\delta}(x)$ where $m=m(\rho\vareps^\chill_\delta(n)/8,1/(n^2 s))$ as in Lemma~\ref{lem:empirical-min-close}, and as in Lemma~\ref{lem:estimate-poly} let $\widehat{q}_{x,\delta}(z)$ denote the polynomial
$\widehat{q}_{x,\delta}(z) = n - \|\widehat{\mu}_{m,\delta}(x)\|_1 + \sum_{j=0}^{n-1} \widehat{\mu}_{m,\delta}(x)_j z^j.$

\item For each $i=1,\dots,s$ the algorithm computes\onote{Not to be a colossal dick, but I suppose we cannot compute it exactly.  Or maybe we sort of can, if we store complex numbers as $(r,\theta)$ pairs.} and records the value $\widehat{q}_{x,\delta}(z_i).$
\ignore{
}
\item Define the polynomial
\begin{equation}
\label{eq:p-in-alg}
p(z)= \rho \cdot \left( n - \|x\|_1 + \sum_{j=0}^{k-1} x_j (\delta + \rho z)^j + \sum_{i=k+1}^{n-1} \alpha_i (\delta + \rho z)^i\right)
\end{equation}
(the degree-$k$ coefficient being 0 corresponds to guessing that $x_k=0$), where $\alpha_{k+1},\dots,\alpha_{n-1}$ are indeterminates.

\item For $i \in [s],$ solve the following convex minimization problem over variables $\alpha_{k+1},\dots,\alpha_{n-1}$:
\begin{align}
&\text{minimize~~~~~~} |\widehat{q}_{x,\delta}(z_i) - p(z_i)|^2 \label{eq:min}\\
&\text{subject to~~~~~} 0 \leq \alpha_i \leq 1 \text{~for~}i=k+1,\dots,n-1, \label{eq:constraints}
\end{align}
to obtain a value $\gamma_i$ that is within $\rho\vareps^\chill_\delta(n)/8$ of the minimum value of
$|\widehat{q}_{x,\delta}(z_i) - p(z_i)|$.  (Note that $|\widehat{q}_{x,\delta}(z_i) - p(z_i)|^2$ can be expressed as $L_1(\alpha_{k+1},\dots,\alpha_{n-1})^2 +
L_2(\alpha_{k+1},\dots,\alpha_{n-1})^2$ for real linear forms $L_1,L_2$, which is indeed a convex function of the $\alpha_i$-variables.)

\item If any $\gamma_i$ is at least $\rho \vareps^\chill_\delta(n)/2$ then output $\widehat{x}_k=1$ (meaning that the algorithm has determined that the guess $x_k=0$ was incorrect); if every $\gamma_i$ is less than $\rho \vareps^\chill_\delta(n)/2$ then output $\widehat{x}_k=0$ (meaning that the algorithm has determined that $x_k=0$ was correct).

\end{enumerate}

The sample complexity of the above algorithm is $m=m(\rho\vareps^\chill_\delta(n)/8,1/(n^2 s))$ and its running time is $\poly(m,s,n,1/\vareps^\chill_\delta(n)) = \poly(n,1/\vareps^\chill_\delta(n))$, both as claimed by Theorem~\ref{thm:alg}.\rnote{Convex minimization is one of my anti-specialities so somebody please check that the things said regarding this are okay.}

It remains to establish correctness.  To analyze correctness, we henceforth assume that the draw of $\widehat{\mu}_{m,\delta}(x) \in [0,1]^n$  from $\widehat{\bmu}_{m,\delta}(x)$ is indeed good as described in the proof of Lemma \ref{lem:empirical-min-close}.

 Let us suppose first that $x_k=0.$   Taking $\alpha_i=x_i \in \{0,1\}$ for $i=k+1,\dots,n-1$, by Lemma~\ref{lem:poly} the polynomial $p(z)$ defined in (\ref{eq:p-in-alg}) is exactly equal to $q_{x,\delta}(z)$, so for each $i \in [s]$ the minimum value of (\ref{eq:min}) is at most $(\rho \vareps^\chill_\delta(n)/8)^2$ by (\ref{eq:qhat-close-to-q}) and hence the obtained value $\gamma_i$ is at most $\rho \vareps^\chill_\delta(n)/4.$  Thus in this case Step~5 correctly outputs $\widehat{x}_k=0.$

The remaining case is that $x_k=1.$  By Lemma~\ref{lem:poly}, any polynomial $p$ as defined in (\ref{eq:p-in-alg}) that is feasible for (\ref{eq:constraints}) has the property that
$
p\left({\frac {z-\delta} \rho}\right) - q\left({\frac {z-\delta} \rho}\right)
$
is $\rho$ times a chill polynomial, so
\[
\max_{z' \in D_\rho(\delta)}
\left|
p\left({\frac {z'-\delta} \rho}\right) - q\left({\frac {z'-\delta} \rho}\right)
\right| \geq \rho \vareps^\chill_\delta(n).
\]
Fix $z'^\ast \in D_\rho(\delta)$ to be such that $\left|
p\left({\frac {z'^\ast-\delta} \rho}\right) - q\left({\frac {z'^\ast-\delta} \rho}\right)
\right| \geq \rho \vareps^\chill_\delta(n),$ and let $i^\ast \in [s]$ be such that $z_{i^\ast}$ is the closest element of $\{z_i\}_{i \in [s]}$ to ${\frac {z'^\ast - \delta} \rho}$ (note that ${\frac {z'^\ast - \delta} \rho}$ lies on the unit circle).  Since $p\left({\frac {z-\delta} \rho}\right) - q\left({\frac {z-\delta} \rho}\right)$ is $\rho$ times a chill polynomial and any chill polynomial's derivative has magnitude at most $n^2$ on every $z' \in D_\rho(\delta)$, we have
\[
\left|
\left(
p\left({\frac {z'^\ast - \delta} \rho}\right) - q_{x,\delta}\left({\frac {z'^\ast - \delta} \rho}\right)
\right) -
\left(
p\left(z_{i^\ast}\right) - q_{x,\delta}\left(z_{i^\ast}\right)
\right)
\right| \leq \rho n^2 \cdot {\frac \pi s} \leq {\frac {\rho \vareps^\chill_\delta(n)} 8},
\]
and hence
\[
\left|
p\left(z_{i^\ast}\right) - q_{x,\delta}\left(z_{i^\ast}\right)
\right| \geq
 {\frac {7\rho \vareps^\chill_\delta(n)} 8}.
\]
By (\ref{eq:qhat-close-to-q}) and the triangle inequality we get that
\[
\left|
p\left(z_{i^\ast}\right) - \widehat{q}_{x,\delta}\left(z_{i^\ast}\right)
\right| \geq
 {\frac {3\rho \vareps^\chill_\delta(n)} 4},
\]
and so the value of $\gamma_{i^\ast}$ must be at least $5\rho \vareps^\chill_\delta(n)/ 8,$ so in Step~5 the algorithm correctly outputs $\widehat{x}_k=1.$
This concludes the proof of Theorem~\ref{thm:alg}.
\qed

}

\ignore{
\newpage

\violet{
\section{Finishing the proof of Theorem~\ref{thm:main-positive}: Lower bound on $\vareps^\chill_\delta(n)$} \label{sec:B}

Notation:  we write $C_R$ to denote the radius-$R$ circle centered at 0 in the complex plane and $C_R(v)$ to denote the radius-$R$ circle centered at $v \in \mathbb{C}$.

For the lower bound on $\vareps^\chill_\delta(n)$ we consider a broader class of polynomials than Littlewood polynomials (this is the class that Ryan adroitly identified in his section below).  These are polynomials  $P(z)=\sum_{j=0}^{n} P_j z^j$ where (i) the smallest $j$ such that $P_j \neq 0$ has
$|P_j|=1$ and (ii) every $|P_j| \leq 1$.  Let's call these ``chill polynomials.'' (I think for the convex programming based algorithm we actually need that every chill polynomial takes a not-too-small value, so that's what we better show here.)
Throughout this section fix $P(z)=\sum_{j=0}^{n} P_j z^j$ to be any chill polynomial.  We write $P(z)$ as $P(z)=z^t \cdot Q(z)$ for some $t \in [0,n]$, where $Q(z)=\sum_{j=0}^{n-t} Q_j z^j$ has $|Q_j| \leq 1$, $Q(0) \in \{\pm 1\}.$

\bigskip

In this section we want to prove that somewhere within the disk $D_{\rho}(1-\rho)$ the value of $|P|$ is ``not too small.''  The high level idea is to do this as follows:  We will define an arc $A_B$ of another (bigger) circle $B$ (to be defined later) such that $A_B$ lies in $D_\rho(1-\rho).$ We will show that $\|Q\|_{A_B}$ (the maximum of $|Q|$ over $A_B$) is ``not too small.'' The value of $|z^t| \geq |z^n|$ also cannot be ``too small'' anywhere on $A_B$ because every point $z$ on $A_B$ has $|z|$ ``very close'' to 1.  Hence at whatever point $z_\star \in A_B$ has $|Q(z_\star)|$  ``not too small'', it will similarly hold that $|P(z_\star)|$ is ``not too small.''

\bigskip

A good way to think about what follows is that $\rho$ and $n$ are fixed and given to us, and we will optimize a parameter $\tau$ (it's an angle) below given these fixed values of $\rho$ and $n$. (We'll see that as long as $\rho = \omega(1/\sqrt{n})$ we will have $\tau=o_n(1)$; we will only consider such $\rho$'s, so we should think of $\tau$ as being $o_n(1)$.)

We are going to be working with three circles in the complex plane.   The first is $U=C_1$ (the unit circle).  The second is $B=C_{1-\gamma}$ (the letter $B$ is for ``big''; $\gamma =\gamma(\tau) \ll \rho$ will be a small value, it is completely determined by $\tau$ and vice versa).  The third is $S=C_{\rho}(1-\rho)$ (the letter $S$ is for ``small'' since we think of $\rho$ as a small value).  $U$ and $S$ share a monogamous kiss at 1. $B$ and $U$ share the center 0 and don't touch anywhere.  Here's a risibly bad picture of the situation:

\begin{center}
\includegraphics[width=150mm, height=100mm]{threecircles}
\end{center}

\bigskip  The angle $\tau$ is defined as follows:  $B$ and $S$ intersect at two points which we write as
\[
z_\tau := (1-\rho) + \rho e^{i \tau}
\quad \text{and} \quad
z_{-\tau} := (1-\rho) + \rho e^{-i \tau},
\]
so these two points $z_\tau,z_{-\tau}$ define an angle of $2\tau$ from the center of $S$ (the point $1-\rho$).  Since $\tau=o_n(1)$ is very small, the region (interior of $S$) $\setminus$ (interior of $B$) is a little crescent-moon-shaped region; the crescent has $z_\tau$ and $z_{-\tau}$ as the two pointy endpoints.  (Observe that indeed if I tell you the value of $\tau$ then that  completely determines the value of $\gamma$ and vice versa as claimed earlier.)

Let's write $A_{S}$ for the arc of $S$ (the longer/rounder part of the crescent) between these two points $z_\tau,z_{-\tau}$, and let's write $A_{B}$ for the arc of $B$ between these two points (the shorter/straighter part of the crescent; this is the part we'll really care about).  Since $\tau$ is small, length$(A_B)$ and length$(A_S)$ are within a constant factor of each other (actually all we need for them to be within a constant factor of each other is that $\tau$ is bounded below $\pi$ by a constant;  in our setting $\tau$ is $o_n(1)$ and the two lengths are within a $1+o_n(1)$ factor of each other) hence length$(A_B)=\Theta(\tau \rho)$ since length$(A_S)$ is exactly $2\tau \rho.$  Since the radius of $B$ is $1-\gamma=\Theta(1)$, the arc $A_B$ is a $\Theta(\tau \rho)$ fraction of the circle $B$.

\bigskip

Starting with the easier part first, we give a lower bound on $|z|^n$ for any point $z \in A_B$:

\begin{claim} \label{claim:a}
Let $z \in A_B$, so $z = (1-\rho) + \rho' e^{i \theta}$ for some $0 \leq \rho' \leq \rho$ and some $|\theta| \leq \tau.$  Then $|z|^2 \stackrel{>}{\sim} e^{-\rho \tau^2},$ and hence $|z|^n \geq e^{-\rho \tau^2 n/2}.$
\end{claim}
\begin{proof}
Since $z$ and $z_\tau$ both lie on $B$, we have $|z|=|z_\tau| = |(1-\rho) + \rho e^{i \tau}| \approx e^{-\rho \tau^2 / 2}$ (by Anindya's handwritten calculations from a little while ago).
\end{proof}

Now we turn to the first part of the argument, showing that $|Q|$ can't be too small everywhere on $A$.  This is the part that uses the word ``harmonic.''

We recall some (hopefully correct) basic definitions and facts.  A function $f: \mathbb{C} \to \R$ is \emph{harmonic} if $g: \R^2 \to \R$ defined by $g(x,y) = f(x+iy)$ is harmonic.   See Wikipedia for the following:

\begin{fact}
Let $Q$ be a holomorphic function on $\mathbb{C}.$  Then $f: \mathbb{C} \to \R$ defined by $f(z)=\ln|Q(z)|$ is harmonic.\rnote{As Anindya astutely pointed out, this is only true if $Q$ is never zero, so we can't use this fact.  We need to use the
Mahler measure stuff as outlined in emails leading up to (and summarized in ) Ryan's 8/25 email.}

\end{fact}

Together with the Wikipedia entry for ``Harmonic function / the mean value property''
this gives us the following:

\begin{fact}
Let $Q(z)$ be a Littlewood polynomial of the form $Q(z) = \pm1 + \sum_{j=1}^n Q_j z^j$ where $Q_j \in \{-1,0,1\}$ ,and let
$f(z)=\ln|Q(z)|.$  Then for the circle $B=C_{1-\gamma}$  of radius $1-\gamma$ centered at 0 in the complex plane $\mathbb{C}$,
we have that the average of $f$ around $B$ equals 0.
\end{fact}
\noindent
(This is because the mean value property for harmonic functions tells us that the average of $f$ around $B$ equals
$f(0)=\ln|(Q(0))| = \ln1 = 0.$)

\bigskip

Recalling from above that $A_B$ is a $\Theta(\tau \rho)$ fraction of the circle $B$, we have

\begin{eqnarray}
0 &= & \Avg_B f(z) = \Avg_B  \ln |Q(z)|\nonumber \\
&=& (1-\Theta(\tau \rho)) \Avg_{B \setminus A_B} \ln|Q(z)| + \Theta(\tau \rho) \Avg_{A_B} \ln|Q(z)|. \label{eq:a}
\end{eqnarray}
Now we observe that for any $z \in B \setminus A_B$, we have $|z| \stackrel{<}{\approx} e^{-\rho \tau^2 / 2} \approx
1 - \rho \tau^2 / 2$, so
\[
|Q(z)| = |\sum_{j=0}^n Q_j z^j| \leq \sum_{j=0}^n |Q_j|\cdot |z^j| \leq \sum_{j=0}^n |z|^j \leq {\frac 1 {1-|z|}}
\approx {\frac 2 {\rho \tau^2}}.
\]
This means that from (\ref{eq:a}) we get (being sloppy about constants)
\[
-C \cdot \ln\left({\frac 1{\rho \tau^2}}\right) \leq \tau \rho \cdot \Avg_{A_B} \ln|Q(z)|,
\]
so
\[
{\frac {-C \cdot \ln\left({\frac 1{\rho \tau^2}}\right)} {\tau \rho}} \leq \Avg_{A_B} \ln|Q(z)|,
\]
and hence there must be some $z_\star \in A_B$ for which
\[
\left(\rho \tau^2\right)^{C/(\tau \rho)}\leq |Q(z_\star)|.
\]
Combining this with Claim \ref{claim:a}, we get that $z_\star$ satisfies
\[
|P(z_\star)| \geq |z_\star|^t \cdot |Q(z_\star)| \geq e^{-\rho \tau^2 n/2} \cdot e^{-\ln(1/(\rho \tau^2)) \cdot C/(\tau \rho)}.
\]
Pretending the $\ln(1/(\rho \tau^2))$ is not there for simplicity for now\rnote{I hear a rumor that Ryan is going to get rid of this factor for us.} we optimize this by choosing $\tau$ so as to
equalize $\rho \tau^2 n$ and $1/(\tau \rho)$. This gives $\tau = {\frac 1 {\rho^{2/3} n^{1/3}}}$ and $|P(z_\star)| \geq e^{-n^{1/3}/\rho^{2/3}}.$  Observe that if $\rho = \omega(1/\sqrt{n})$ then $\tau=o_n(1)$ and our claims from earlier about $\tau$ indeed hold.

\subsection{Addition by Ryan just before leaving SpB}

\newcommand{\Exp}{\mathop{\mathrm{Exp}}}
Yeah, I think I can get rid of the log factor.  See if you dig this.  I'll assume WLOG that $\rho \leq 1/4$, as this only hurts. (Do I need this?  Not sure.)  Let me also invent a new notation: $\Exp(\cdot)$ means $\exp(O(\cdot))$.\\

As before, let $P(z)$ be any chill polynomial, i.e. any polynomial expressible as $z^dQ(z)$, where $Q$'s constant coefficient has magnitude~$1$ and all its other coefficients have magnitude at most~$1$.  Note that we don't assume any bound on the degree of~$Q$.  (Does this let us solve any slight generalization of the trace-reconstruction problem, by the way?)  As always, we get that $Q(|z|) \leq 1/\vareps$ whenever $|z| \leq 1-\vareps$.  We also easily get that $|Q(1/3)| \geq 1/2$.  Let $f(z) = \log_2|Q(z)|$, so $f$ is harmonic, $f(1/3) \geq -1$, and $f(z) \leq \log(1/\vareps)$ whenever $|z| \leq 1-\vareps$.

\begin{lemma} \label{lem}
    Let $0 \leq \tau \leq \pi/2$, let $s = 1-\rho + \rho e^{i \tau}$,  let $G$ be the circle with center $1/3$ passing through $s, \overline{s}$, and let $L$ be the minor arc of~$G$ between $s, \overline{s}$.  Then the average value of $|Q|$ along~$L$ is at least $\Exp(-1/(\rho \tau))$.
\end{lemma}
With this lemma in hand, we finish as before:  Note that $|s| = 1 - \Theta(\rho \tau^2)$ and hence $|z| \geq 1 - O(\rho \tau^2) \geq \Exp(-\rho \tau^2)$ for all points $z \in L$. Thus the average of $|P|$ along $L$ is at least $\Exp(-\rho \tau^2 d - 1/(\rho \tau))$.  Taking $\tau = \Theta(1/(d^{1/3} \rho^{2/3}))$ we get:
\begin{corollary}
    The average value of $|P(z)|$ along $z \in L \subset D_\rho(1-\rho)$ is at least $\Exp(-d^{1/3}/\rho^{1/3})$.
\end{corollary}
\begin{proof} (Lemma~\ref{lem}.)  Since $f$ is harmonic, the average value of~$f$ around~$G$ is at least $-1$.  Now $G$ is the locus of points $1/3 + r e^{i\sigma}$ for $\sigma \in (-\pi, \pi]$, where $r = |s - 1/3| \leq 2/3$.  The modulus of such a point is $1 - \Omega(\sigma^2)$; hence $Q$'s modulus at such a point is $O(1/\sigma^2)$ and $f$'s value at such a point is at most $\log(O(1)/\sigma)$.  Let $L'$ be the big complementary arc to $L$ in~$G$, given by $|\sigma| \geq \sigma_0$.  It's not going to be relevant what $\sigma_0$ is, as a function of $\rho$ and $\tau$, although I guess it's $\Theta(\rho \tau)$.  Anyway, the average of $f$ along~$L'$ is therefore
\[
    \Theta(1) \int_{\sigma_0}^\pi \log(O(1)/\sigma)d\sigma \leq O(1) \int_0^\pi \log(O(1)/\sigma ) d\sigma \leq O(1).
\]
It follows that the average of $f$ along $L$ is at least $-O(1)/|L| \geq -O(1)/(\rho \tau)$.  Thus
\[
    -O(1)/(\rho \tau) \leq \E_{z \sim L} [\log |Q(z)|] \leq \log \E_{z \sim L} [|Q(z)|],
\]
as claimed.
\end{proof}

}

}

\ignore{
\newpage
\violet{
\section{Finishing the proof of Theorem~\ref{thm:main-negative}:  Upper bound on $\vareps^\Lit_\delta(n)$} \label{sec:A}

Recall that our goal now is to show that there is a Littlewood polynomial $P_{\text{loser}}(z)$ such that $|P_{\text{loser}}(z)| \leq e^{-n^{1/3}/\rho^{1/3}}$ for all $z \in D_\rho(1-\rho)$.  (In this section we actualy need $P_{\text{loser}}$ to be Littlewood, it's not enough for it to just be chill.)
The key tools here are Corollary 4.5 and Lemma 4.6 of the Borwein/Erdelyi paper.  The former is evidently an easy consequence of the Hadamard Three Circles Theorem and the latter is proved in another B/E paper.\rnote{These are the main ingredients of Theorem 3.3 of the B/E paper; however right now it seems to me it may be slightly more convenient to just use these instead of using Theorem 3.3 per se.}  At a high level we will use these tools to assert the existence of a Littlewood polynomial $P_{\text{loser}}(z)$ which has the property that for every $z \in S$, $|P_{\text{loser}}(z)| \leq e^{-n^{1/3}/\rho^{1/3}}$ (the maximum modulus principle tells us that if $|P_{\text{loser}}(z)| \leq e^{-n^{1/3}/\rho^{1/3}}$ for every $z \in S$ then $|P_{\text{loser}}(z)| \leq e^{-n^{1/3}/\rho^{1/3}}$ for every $z \in D_\rho(1-\rho)$ which is our final goal).

In more detail, the polynomial $P_{\text{loser}}(z)$ will be
$z^n \cdot Q_{\text{loser}}(z)$ where $Q_{\text{loser}}(z)$ is a degree-$n$ Littlewood polynomial (so really $P_{\text{loser}}$ may have
 degree as high as $2n$ but really, what is the difference between friends).  To argue that $|P_{\text{loser}}(z)| \leq e^{-n^{1/3}/\rho^{1/3}}$ for
every $z \in S$, we will define a certain ellipse $\tilde{E}_a$ and an arc $A_{\tilde{E}_a}$ of $S$ that lies within that ellipse and contains
 the point 1.  Corollary 4.5 and Lemma 4.6 will together give us that there is a Littlewood polynomial $Q_{\text{loser}}(z)$ such that $
 |Q_{\text{loser}}(z)|$ is small everywhere on the ellipse $\tilde{E}_a$, hence (by the maximum modulus principle) small everywhere within the ellipse, hence small everywhere on $A_{\tilde{E}_a}$ (and hence so is $|P_{\text{loser}}(z)| = |z|^n |Q_{\text{loser}}(z)| \leq |Q_{\text{loser}}(z)|$, where these inequalities hold on all $z \in A_{\tilde{E}_a}$ since such $z$ have $|z| \leq 1$).   Because all points $z \in S \setminus A_{\tilde{E}_a}$ have $|z|$ bounded away from 1, all points $z \in  S \setminus
 A_{\tilde{E}_a}$ will have that $|P_{\text{loser}}(z)| = |z|^n \cdot |Q_{\text{loser}}(z)| < n \cdot |z|^n$ is small.  Thus all points on $S$ will have
 that $|P_{\text{loser}}(z)|$ is small as desired.

 \begin{center}
\includegraphics[width=150mm, height=100mm]{ellipse}
\end{center}

So let's get down to it.  As before think of $\rho,n$ as fixed where $\rho = \omega(1/\sqrt{n})$.  Let $\tilde{E}_a$ be as defined in Corollary 4.4 (the foci should be $1-8a$ and 1, not $1-a$ and 1 -- that's a typo -- and the major axis is $[1-14a,1+6a]$) where $a<\rho,a=o_n(1)$ is a parameter we will optimize later.

Lemma 4.6 tells us there is a Littlewood polynomial $Q_{\text{loser}}(z)$ of degree $d$  that has magnitude at most $e^{-c_5 \sqrt{d}}$ for every real $z \in [0,1]$, so certainly $|Q_{\text{loser}}(z)| \leq e^{-c_5 \sqrt{d}}$ for every $z \in [1-8a,1].$  We will take $d=\Theta(1/a^2$) (and will check/stipulate later that $a = \Omega(1/\sqrt{n})$ so that $d \leq n$), so we have
$|Q_{\text{loser}}(z) \leq e^{-c_5/a}$.  Corollary 4.5 tells us that we have
\[
\max_{z \in \tilde{E}_a} |Q_{\text{loser}}(z)| \leq \Theta(1/a) \cdot e^{13da/2} \cdot e^{-c_5 /a} = e^{-\Theta(1/a)}
\]
(for a suitable choice of the hidden constant in $d=\Theta(1/a^2)$,
so we get
\begin{equation}
\label{eq:first}
\max_{z \in A_{\tilde{E}_a}} |P_{\text{loser}}(z)| \leq
\max_{z \in A_{\tilde{E}_a}} |Q_{\text{loser}}(z)| \leq
\max_{z \in \tilde{E}_a} |Q_{\text{loser}}(z)| \leq
e^{-\Theta(1/a)},
\end{equation}
where the first inequality is because $|P_{\text{loser}}(z)| \leq |Q_{\text{loser}}(z)|$ for every $|z| \leq 1$ and the second is
the maximum modulus principle.

To bound $|P_{\text{loser}}(z)|$ for $z \in S \setminus A_{\tilde{E}_a}$, we observe that the arc $A_{\tilde{E}_a}$ has length $\Theta(a)$.  Hence the two endpoints of $A_{\tilde{E}_a}$ are at points
\[
z_{\tau'} := (1-\rho) + \rho e^{i \tau'}
\quad \text{and} \quad
z_{-\tau'} := (1-\rho) + \rho e^{-i \tau'}
\]
for some angle $\tau' = \Theta(a/\rho).$  We have $|z_{\tau'}| \approx e^{-\rho (\tau')^2/2}$, so any $z \in S \setminus A_{\tilde{E}_a}$
has $|z| \leq e^{-\rho (\tau')^2/2}$ and hence
\begin{equation}
\label{eq:second}
\max_{z \in S \setminus  A_{\tilde{E}_a}} |P_{\text{loser}}(z)| \leq n |z|^n \leq n e^{-\rho (\tau')^2n/2} =
n e^{-\Theta(a^2 n/\rho)}.
\end{equation}
Ignoring the ``$n$'' out front (it's late, and it is ignorable) and some constants, we minimize the max of these two upper bounds (\ref{eq:first})
and (\ref{eq:second}) by taking $\Theta(1/a)=a^2 n/\rho$, i.e. $a=\Theta(\rho^{1/3}/n^{1/3})$, and we get
\[
\max_{z \in S} |P_{\text{loser}}(z)| \leq e^{-\Theta(n^{1/3}/\rho^{1/3})}.
\]
Note that if $\rho = \omega(1/\sqrt{n})$ then $a = \omega(1/\sqrt{n})$ and we have $a<\rho$, $d \leq n$ as required above.

\ignore{
  We use the following variant of Theorem 3.3 of B/E (we call this Theorem 3.3$'$): \emph{``Let $0 < a < $(some constant).  Let $A$ be the subarc of $S$ with angular measure $a$ centered at the point 1.  Then there is a Littlewood polynomial $Q_{\text{loser}}(z)$ of degree $\Theta(1/a^2)$ such that
$|Q_{\text{loser}}(z)| \leq e^{-1/a}$ for all $z \in A.$''}  (This can be verified from the first proof of Theorem 3.3 given on page 11 of B/E.  The only thing that changes is the next to last sentence where now we must observe that since the ellipse $\tilde{E}_a$ with foci $1-8a$ and 1 (see Corollary 4.4 -- the ``$1-a$'' there is a typo for ``$1-8a$'') and major axis $[1-14a,1+6a]$ has constant aspect ratio, as claimed there the unit circle intersects $\tilde{E}_a$ in an arc of length $\Omega(a)$.)

}

}
}

\ignore{

\section{Other results}

In \cite{MoitraSaks13}, Moitra and Saks gave an $(n/\eps)^{2f(\eta) + O(1)}$ time algorithm for the ``population recovery problem with erasure probability $1-\eta$'', where $f(\eta) = \log(2/\eta)/\eta$. \red{(Recall that this problem is to reconstruct an unknown distribution over $\{0,1\}^n$ to within an additive error $\eps$ in the probability of each point, given access to independent draws from the distribution in which each coordinate of each draw is independently replaced with ``?'' with probability $1-\mu$; see \cite{MoitraSaks13} for details.)}   They showed (see Section~2 of their paper) that this result is a direct consequence of proving that $\sigma(\eta,\ell) \leq \ell^{f(\eta)}$, where $\sigma(\eta,\ell)$ is the ``minimum sensitivity of a $1/\ell$-local inverse'' (see their Section~1.2).  In turn, they established this upper bound on $\sigma(\eta,\ell)$ in two steps.  First they showed (see their Sections~3 and~4) that $\sigma(\eta,\ell)$ is at most
\begin{equation}    \label{eqn:LP}
    \max \braces*{P(0) - \frac{L(P)}{\ell} : P \text{ is a real-coefficient polynomial with} \max_{D_\eta(1-\eta)} |P| \leq 1};
\end{equation}
here $L(P)$ denotes the \emph{length} of~$P$, the sum of the magnitudes of its coefficients.\footnote{Actually, they showed a stronger bound in which one minimizes over $P$ where the length of $P(1-\eta + \eta z)$ (with respect to variable~$z$) is at most~$1$.  However they also showed that this length is always at least $\max_{D_\eta(1-\eta)} |P| \leq 1$.}  Then they showed that~\eqref{eqn:LP}  is upper-bounded by $\ell^{\log(2/\eta)/\eta}$ (see their Section~5).

Here we improve the upper bound on \eqref{eqn:LP} to~$\ell^{g(\eta)}$, where $g(\eta) = \pi/\sin^{-1}\frac{\eta}{1-\eta} - 1 \leq \pi/\eta$. Tracing through the analysis of \cite{MoitraSaks13}, this yields an  $(n/\eps)^{2g(\eta) + O(1)} = (n/\eps)^{O(1/\eta)}$ time algorithm for the population recovery problem, improving their exponent by a logarithmic factor.

To establish this bound, let us first simplify the form of~\eqref{eqn:LP}.  Certainly the maximum is nonnegative, which means that we can freely add $P(0) \geq L(P)/\ell$ as an additional constraint.  Next, it only hurts us to change the objective function to $P(0)$, rather than $P(0) - L(P)/\ell$.  Next, the optimum is certainly achieved when $\max_{D_\eta(1-\eta)} |P| = 1$, so we can change to this constraint.  In turn we can now replace the objective function $P(0)$ by $P(0)/\max_{D_\eta(1-\eta)} |P|$.  The constraint $\max_{D_\eta(1-\eta)} |P| = 1$ is now redundant, as the objective function and the constraint $P(0) \geq L(P)/\ell$ are homogeneous.  Finally, by this homogeneity we can add the constraint $P(0) = 1$.  To summarize, we have shown that
\[
    \eqref{eqn:LP} \leq \max \braces*{\frac{1}{\max_{D_\eta(1-\eta)} |P|} : P(0) = 1, L(P) \leq \ell} = \parens*{\min  \braces*{\max_{D_\eta(1-\eta)} |P| : P(0) = 1, L(P) \leq \ell} }^{-1}.
\]
Thus to obtain our improved population recovery result, it suffices to show that if $P$ is a polynomial with $P(0) = 1$ and $L(P) \leq \ell$, then $\max_{D_\eta(1-\eta)} |P| \geq \ell^{-g(\eta)}$.

This can be done using a simplified version of our Theorem~\ref{thm:BE-like}.  \red{XXXXXXtoo lazy to write it just now; but seriously, just pass an origin-centered circle thru $\bdry D_\eta(1-\eta)$ with the intersection of the two circles being at right-angles.  Then apply the Mahler Measure lemma.  Done.XXXXXXX}

\onote{The $g(\eta) = \pi/\eta$ bound is is tight up to a constant factor.  Do we want to bother putting in the argument?  I wrote it in a June 1 email}

}

\section{Conclusions}
A natural direction for future work is to go beyond mean-based algorithms.  For example, an efficient algorithm can estimate the covariances of all \emph{pairs} of trace bits.  If different sources strings lead to sufficiently different trace-covariances, one could potentially get a more efficient trace reconstruction algorithm.  Analyzing this strategy is  equivalent to analyzing a certain problem concerning the maxima of Littlewood-like polynomials on~$\C^2$; however we could not make any progress on this problem.  It would also be interesting to develop lower bound techniques that apply to a broader class of algorithms than just mean-based algorithms.

\ignore{
Another avenue for progress might be to abandon simple statistical algorithms altogether.  For example, the Bitwise Majority Alignment variant of Batu~et~al.~\cite{BKKM04} efficiently solves trace reconstruction when $\delta = 1/n^{.51}$; this does not contradict our lower bound because it is not a mean-based algorithm.
}

Finally, we mention that the authors have applied the techniques in this paper (specifically, the technique used in Section~\ref{sec:improvement}) to several aspects of the population recovery problem.  Details will appear in a forthcoming work.

\subsection*{Acknowledgments}
The authors would like to thanks the Simons Foundation for sponsoring the symposium on analysis of Boolean functions where the authors began work on this project.
A.~D.~would like to thank Aravindan Vijayaraghavan for useful discussions about this problem.

\bibliography{allrefs}{}
\bibliographystyle{alpha}

\appendix

\section{Results on channels that allow insertions, deletions and flips} \label{app:general}

\subsection{Defining the general channel} \label{sec:channel-definition}
We now describe the most general channel $\channel$ that we analyze, which we subsequently refer to as ``the general channel''.  As stated earlier, this channel allows for three different types of corruptions: deletions with probability~$\delta$,  insertions with probability~$\sigma$, and bit-flips with probability~$\gamma/2$.  We comment that for mean-based algorithms, the presence of bit-flips makes hardly any difference; thus the reader may focus just on the combination of deletions and insertions.

Our definition of this general channel is essentially the same as that of Kannan and McGregor~\cite{KM05}.  More precisely, for parameters $\delta, \sigma, \gamma \in [0,1)$, we define how the channel acts on a single source bit~$b \in \{-1,1\}$:
\begin{enumerate}
	\item First, the channel performs ``insertions''; i.e., it repeatedly does the operation ``with probability~$\sigma$, transmit a uniformly random bit; with probability $1-\sigma$, stop''.
    \item Having stopped, the channel ``deletes'' (completes transmission without sending $b$ or $-b$) with probability~$\delta$.
    \item Otherwise (with probability~$1-\delta$), the channel transmits one more bit:  namely,~$b$ with probability $1-\gamma/2$, or $-b$ with probability $\gamma/2$.
\end{enumerate}
As usual, the channel operates on an entire source string $x \in \{-1,1\}^n$ by operating on its individual bits independently, concatenating the results.  That is,
\[
	\channel(x) = \channel(x_0)\channel(x_1) \cdots \channel(x_{n-1}) \in \{-1,1\}^*.
\]
Of course, if we set $\sigma = \gamma = 0$, we get the deletion channel $\Del_\delta$ that was analyzed in the main body of the paper. 

An alternative description of the channel's operation on a single bit~$x_i$ is as follows:
\begin{equation} 
\calC (x_i) = \begin{cases} \bw \ &\textrm{with probability }\delta, \\
(\bw ,   \ba) \ &\textrm{with probability } (1-\delta) \cdot \gamma, \\
(\bw ,   x_i) \ &\textrm{with probability } (1-\delta) \cdot (1-\gamma),  \end{cases} \label{eq:cw}
\end{equation}
where $\ba \in \{-1,1\}$ is a uniformly random bit, and where $\bw \in \{-1,1\}^{\bG}$ is a uniformly random string of~$\bG$ bits, with $\bG$ in turn being a Geometric random variable of parameter~$1-\sigma$.\footnote{ Here we use the convention that Geometric random variables take values $0, 1, 2, \dots$ (equal to the number of ``failures''); i.e., $\Pr[\bG = t] = \sigma^t(1-\sigma)$ for each $t \geq 0$.}  From this description one can see that in a received word $\by \leftarrow \channel(x)$, each received bit either ``comes from a properly transmitted source bit~$x_i$'', or else is uniformly random.  (The probability each $x_i$ comes through is $(1-\delta)(1-\gamma)$.)  As a consequence, we have that Proposition~\ref{prop:its-linear} continues to hold for~$\channel$: for every~$j \in \N$, the mean value $\E_{\by \leftarrow \channel(x)}[\by_j]$ is a \mbox{(real-)linear} function of~$x$.

Note that when the insertion probability $\sigma$ is positive, the received word $\by \leftarrow \channel(x)$ does not have an a priori bounded length.  This is a minor annoyance can be handled in several different ways; we choose one way in the next section.

\ignore{
Thus $\delta$ corresponds to the deletion probability, $\sigma$ to the insertion probability (where each inserted bit is uniform random), and ${\frac {\gamma} 2}$ to the flip probability.  \ignore{\rnote{This is the motivation for (\ref{eq:cw}) and for having the Geometric parameter be $1-\sigma$:  we want it to be the case that setting $\sigma$ and $\gamma$ to 0, we recover the $\Del_\delta$ channel.}}  

We adopt the following notation: for a nonempty string $z \in \{-1,1\}^{\ast},$ let $\mathsf{Del}(z)$ denote
the string obtained by deleting the leftmost position of $z$. As an aid to understanding the general channel, we give an alternate definition  which results in the same distribution $\calC (x)$  is as follows: Start with a \red{nonempty} string $y = x$
and repeat the following process:
\begin{enumerate}

\item Sample independent Bernoulli random variables $\bB_1, \bB_2, \bB_3 \in \{0,1\}$ with $\Pr[\bB_1=0]=\sigma$, $\Pr[\bB_2=0] = \delta$ and $\Pr[\bB_3=0] = (2-\gamma)/2$.

\item If $\bB_1=0$, then transmit a random bit and go to Step~1.

\item If \ignore{$\bB_1=1$ and }$\bB_2=0$, then set $y \leftarrow \mathsf{Del}(y)$ and go to Step~\red{5}.

\item If \ignore{$\bB_2=1$ and }$\bB_3=0$, then transmit $y_1$, \red{else transmit $-y_1.$}  Set $y \leftarrow \mathsf{Del}(y)$ and go to Step~5.

\item \red{If $y$ is the empty string then halt, else go to Step~1.}

\end{enumerate}
}



\ignore{
As in Section~\ref{sec:prelim} we write $\by \leftarrow \channel(x)$ to denote that $\by = (\by_0, \by_1,  \dots, \by_{\boldn - 1})$ is a random received string obtained by passing $x$ through the general channel $\channel.$  For the general channel, if $\sigma >0$ then the length $\boldn$ of the received string may be arbitrarily large with nonzero probability. This is a minor annoyance that can be handled in several different ways; we handle it as follows.}

\subsection{Mean traces for the general channel}

\ignore{
  \begin{assumption}\label{ass:channel}
       \ignore{\item \label{item:N} } There is a \emph{trace length bound}~$N$ associated with\ignore{\onote{We don't \emph{really} need the bound $N \leq n^{O(1)}$ but it makes some things easier to state.  If you're truly sad that this prevents us from handling subpolynomial insertion channel parameters\dots then I feel sorry for you :)}} the general channel, with $n \leq N \leq \poly(n)$, such that ${\Pr[\boldn > N] < 2^{-n}}$.\footnote{The choice of $2^{-n}$ is somewhat arbitrary here.}  \ignore{ When $\channel$ is a deletion channel, we take $N = n$ (and then $\Pr[\boldn > N] = 0$).}
\end{assumption}
Assumption~\ref{ass:channel} holds as long as $\red{\sigma}$ is at most $1 - n^{-\Omega(1)}$. (The alternative regime of $\red{\sigma} = 1-n^{-o(1)}$ is quite extreme, as it means that a super-polynomial number of random bits are typically inserted between every two consecutive source bits; thus this assumption loses little generality.)
\ignore{Assumption~\ref{item:N} is mainly for the handling of insertion channels, which do not have an a priori upper bound on~$\boldn$.}
}

We revisit some of our definitions and observations about mean traces from Section~\ref{sec:mean}, 
in our new context of the general channel. We begin with~\eqref{eq:meantrace}, the definition of the mean trace.  Since the length of a received word may now be arbitrarily large, the mean trace is now an infinite vector.  We deal with this by truncating it at what we call the ``effective trace length bound~$N$''.

\begin{definition} \label{def:tracelengthbound}
For the general channel $\channel$ with insertion probability $0 \leq \sigma < 1$, we define the \emph{effective trace length bound} $N=N(\sigma)$ to be {$N = \left \lceil
10 \cdot {\frac {n +  \ln(1/(1-\sigma))}{1-\sigma}}\right \rceil \leq \poly(n, \frac{1}{1-\sigma})$.}
\end{definition}

\begin{definition} \label{def:deletion-mean-trace}
    For the general channel~$\channel$ and a source string $x \in \{-1,1\}^n$, we define the \emph{idealized mean trace} to be the infinite sequence
    \[
        \mu^{\ideal}_\channel(x) = \E_{\by \leftarrow \channel(x)}[(\by, 0, 0, 0, \dots) ] \in [-1,+1]^\N.
    \]
    We define just the \emph{mean trace} to be its truncation to length~$N$:
    \[
        \mu_\channel(x) = (\mu^{\ideal}_\channel(x)_0, \mu^{\ideal}_\channel(x)_1, \dots, \mu^{\ideal}_\channel(x)_{N-1}) \in [-1,+1]^N.
    \]
\end{definition}

Recalling (\ref{eq:cw}), we see that the length $\boldn$ of a received word is stochastically dominated by $(\bG_1 +1) + \cdots + (\bG_n + 1)$, where the $\bG_i$'s are i.i.d.\ random variables distributed as $\mathrm{Geometric}(1-\sigma)$.  We upper bound this using Janson's bound on the sum of independent Geometric random variables (Theorem~2.1 of \cite{Janson:TB}), noting that his Geometric random variables count the number of ``trials'', which aligns precisely with our $(\bG_i + 1)$'s.  His bound 
gives that
$
\Pr[\boldn \geq N + j] \leq \exp(-(N+j)(1-\sigma)/2)
$ for any $j \geq 0$, and hence we have the following: 
\ignore{
}
for any $x \in [-1,1]^n$,
\begin{align} 
\| \mu_\channel(x) - \mu^{\ideal}_\channel(x)\|_1
&=
\sum_{\ell=N}^\infty |\mu^{\ideal}_\channel(x)|_\ell
\leq
\sum_{\ell=N}^\infty \Pr[\boldn \geq \ell]
=
\sum_{j=0}^\infty \Pr[\boldn \geq N+j]\nonumber \\
&= {\exp(-N(1-\sigma)/2) \cdot  {\frac 1 {1 - \exp(-(1-\sigma)/2)}}
< {\frac {4\exp(-N(1-\sigma)/2)} {1-\sigma}}
}\nonumber \\
&\leq 4 \exp(-n), \quad \quad \quad \quad \text{by our choice of $N$}.
\label{eq:negligible}
\end{align}

\pparagraph{The mean-based trace reconstruction model for the general channel.}
Definition~\ref{def:meanbased}  has a natural analogue for the general channel: an algorithm in the mean-based general-channel model specifies a cost parameter $T \in \N$ and is given an estimate $\widehat{\mu}_\channel(x) \in [-1,1]^N$ of the mean trace satisfying
$        \|\wh{\mu}_{\channel}(x) - \mu_{\channel}(x)\|_1 \leq 1/T.$
It is clear that an algorithm in the mean-based general-channel trace reconstruction model with cost $T_1$ and postprocessing time $T_2$ may be converted into a normal trace reconstruction algorithm using $\poly(N,T_1)=\poly(n,{\frac 1 {1-\sigma},T_1})$ samples and $\poly(n,{\frac 1 {1-\sigma},T_1}) + T_2$ time.  Note that since we will be studying algorithms with cost $T \ll 2^n$, by~(\ref{eq:negligible}) there is no real difference between getting an estimate of~$\mu_{\channel}(x)$ or of $\mu^{\ideal}_{\channel}(x)$.

\pparagraph{The complexity of mean-based trace reconstruction for the general channel.} Regarding the complexity of mean-based trace reconstruction, for the general channel we define $\gap_{\channel}(n)$ and $\gap^{\chill}_{\channel}(n)$ in the obvious way, replacing each occurrence of the length-$n$ vector $\mu_{\Del_\delta}(\cdot)$ in Definition~\ref{def:gap} with the length-$N$ vector $\mu_{\channel}(\cdot).$  \red{As in Section~\ref{sec:complexity}, to show that trace reconstruction can be performed under the general channel in time $\poly(N,M)=\poly(n,{\frac 1 {1-\sigma}},M)$ it suffices to show that $\gap^{\chill}_{\channel}(n) \geq 1/M$}.\footnote{Again, to carry out the linear-programming algorithm, we can either assume that the channel parameters $\delta$, $\sigma$, $\gamma$ are known to the algorithm, or else they should estimated; we omit the details here.}

\pparagraph{Reduction to complex analysis for the general channel.} For $x \in \{-1,1\}^n$ the \emph{general-channel polynomial} is defined entirely analogously to Definition~\ref{def:P}:
    \[
            P_{\channel,x}(z) = \sum_{j < N} \mu_{\channel}(x)_j \cdot z^j;
    \]
note that this is a polynomial of degree less than $N$.  This definition extends to $x \in [-1,+1]^n$ using the linearity of $\mu_{\channel}$.  Similarly, we may define the \emph{idealized general-channel ``polynomial''} by
    \[
        P^{\ideal}_{\channel,x}(z) = \sum_{j \in \N} \mu^{\ideal}_\channel(x)_j \cdot z^j;
    \]
this will actually be a rational function of~$z$.

Entirely analogous to Proposition~\ref{prop:cx}, we get that for every $b \in [-1,1]^n,$
\[
\max_{z \in \bdry D_1(0)} \abs*{P_{\channel,b}(z)} \leq \|\mu_\channel(b)\|_1 \leq \sqrt{N} \max_{z \in \bdry D_1(0)} \abs*{P_{\channel,b}(z)}.
\]

Similar to Section \ref{sec:reduc}, a factor of $\sqrt{N} = \poly(n, {\frac 1 {1-\sigma}})$ is negligible compared to the bounds we will prove, so it suffices to analyze  $\max_{z \in \bdry D_1(0)} \abs*{P_{\channel,b}(z)}$ rather than $\|\mu_\channel(b)\|_1$ in the definitions of $\gap_\channel(n)$ and $\gap^{\chill}_\channel(n)$.
\red{Moreover, since by (\ref{eq:negligible}) we have that $|P^{\ideal}_{\channel,b}(z) - P_{\channel,b}(z)| \leq 2^{-n}$ for all $b  \in [-1,1]^n$ and all $z \in \partial D_1(0)$, it suffices to analyze  $\max_{z \in \bdry D_1(0)} \abs*{P^{\ideal}_{\channel,b}(z)}$; we do this in the next subsection.}

\subsection{Channel polynomial for general channels} \label{sec:general-channels}

We now compute the ideal channel polynomial for the general channel defined in Section~\ref{sec:channel-definition}, using the same technique as in Section~\ref{sec:channel-polynomial-delete} and recalling the discussion around the alternative channel description~\eqref{eq:cw}.  As usual, let $\rho = 1-\delta$.  Let $\bJ_i$ be the random variable whose value is $\bot$ if~$x_i$ is either deleted (probability~$\delta$) or is replaced by a random bit (probability~$(1-\delta)\cdot\gamma$), or else is the position~$j$ such that coordinate~$x_i$ of the source string ends up in coordinate $j$ in the received string $\by$.  As before we let $\wt{\bJ}_i$ denote the random variable $\bJ_i$ conditioned on not being~$\bot$.  Since $\Pr[\bJ_i \neq \bot] = (1-\delta) \cdot (1-\gamma) $, a derivation identical to that of~\eqref{eqn:P-formula} yields
%
\begin{equation}
	P^{\ideal}_{\channel,x}(z)
= (1-\delta) (1-\gamma) \sum_{i < n}  x_i  \cdot \E[z^{\wt{\bJ}_i}]. \label{eq:pideal}
\end{equation}

To compute $\E[z^{\wt{\bJ}_i}]$, it is straightforward to see that each coordinate $x_{i'}$ with $i' < i$ independently generates a random number of received positions distributed as $\bG +\bB$, where  $\bG \sim \mathrm{Geometric}(\red{1-\sigma})$ and independently $\bB \sim \mathrm{Bernoulli}(\rho)$. Further, conditioned on $x_i$ not being deleted, $x_i$ generates a number of received positions distributed as $\bG+1$, \red{where the final ``$+1$'' is for $x_i$ (or $-x_i$) itself.}
 Thus $\wt{\bJ}_i$ is distributed as
 \[
 	\bG_0 + \cdots + \bG_{i} + \bB_0 + \cdots + \bB_{i-1},
 \]
 where the $\bG_k$'s are independent copies of $\bG$ and the $\bB_k$'s are independent copies of $\bB$. We therefore obtain
 \[
 	\E[z^{\wt{\bJ}_i}] = \E[z^{\bG}]^{i+1} \cdot \E[z^{\bB}]^{i} = \parens*{\E[z^{\bG}] \cdot \E[z^{\bB}]}^{i} \cdot \E[z^{\bG}].
 \]
Let $F_G(z)$ denote $\E[z^{\bG}]$ and let $F_B(z)$ denote $\E[z^{\bB}]$.  It is easy to calculate that
$F_G(z) = \red{\frac{1-\sigma}{1-\sigma z}}$, and we saw earlier that $F_B(z) = (1-\rho) + \rho z = \delta + \rho z$.  For brevity, let us write
 \[
 	 w = F_G(z)F_B(z) = \frac{(1-\sigma) \cdot (\delta + \rho z)}{1- \sigma z}, \ignore{= {\frac {(1-\sigma)(1-\rho) + (1-\sigma)\rho z}{1 - \sigma z}}
     }
 \]
\ignore{
\rnote{The expression for $w$ earlier was $\frac{\sigma \cdot ((1-\rho) + \rho z)}{1-(1-\sigma)z}$, before the $\sigma \to 1-\sigma$ switch.}
\ignore{Then, note that

\rnote{Before the $k \to 1-\sigma$ switch what was written was  \[
\gray{F_G(z) = \frac{\rho + (1-\sigma)(1-\rho)}{(1-\sigma) w + \sigma \rho}.}
\]  After our swap of $1-\sigma$ for $\sigma$, this would be
\[
F_G(z) = \frac{\rho + \sigma(1-\rho)}{\sigma w + (1-\sigma) \rho}
\]
but that is not what I got when I solved for $F_G(z)$ in terms of $w$, instead I got what is in blue --- the reciprocal of this.
}
\blue{
\[
F_G(z)=\frac{\sigma w + (1-\sigma) \rho}{\rho + \sigma(1-\rho)}.
\]}}
}
which is a M\"obius transformation of~$z$. Thus~$w$ ranges over a complex circle as $z$ ranges over $\partial D_1(0)$. More specifically, as $z$ ranges over $\partial D_1(0)$ we have that $w$ ranges over $\partial D_r(1-r)$, where 
 \[
r = {\frac {\rho + \delta \sigma}
{1+\sigma}}.
\]
Plugging this back into (\ref{eq:pideal}) using $\E[z^{\tilde{\bJ}_i}]=F_G(z) \cdot w^i,$ we obtain
 \[
 P^{\ideal}_{\channel,x}(z) =
(1-\delta) \cdot (1-\gamma) \cdot F_G(z)\ignore{\frac{\sigma w + (1-\sigma) \rho}{\rho + \sigma(1-\rho)}} \cdot \sum_{i<n}  x_i  \cdot w^i =
(1-\gamma)\cdot (1-\delta) \cdot {\frac {1-\sigma}{1-\sigma z}} \cdot \sum_{i<n}  x_i  \cdot w^i.
 \]


We use the bound $\abs*{{\frac {1-\sigma}{1-\sigma z}}} \geq {\frac {1-\sigma} 2} $ for $z \in \partial D_1(0)$. Now by the analysis of $\kappa^{\chill}_{\bounded}(r,d)$ given in Section~\ref{sec:channel-polynomial-delete} we get the following algorithmic result for general-channel trace reconstruction, which is our most general positive result:

\medskip
\noindent {\bf Theorem~\ref{thm:main-positive-general}, restated.} 
\emph{Let $\channel$ be the general channel described in Section~\ref{sec:channel-definition} with  deletion probability $\delta = 1-\rho$, insertion probability $\sigma$, and bit-flip probability $\gamma/2$.  Define
\[
r \coloneqq {\frac {\rho + \delta \sigma}{1+\sigma}}.
\]
Then there is an algorithm for $\channel$-channel trace reconstruction using samples and running time bounded by
\[
	\poly(\tfrac{1}{1-\delta}, \tfrac{1}{1-\sigma},\tfrac{1}{1-\gamma}) \cdot \begin{cases}
    	\exp(O(n/r)^{1/3}) & \text{if $C/n^{1/2} \leq r \leq 1/2$,} \\
        \exp(O((1-r) n)^{1/3}) & \text{if $O(\log^3 n)/n \leq 1-r \leq 1/2$.}
    \end{cases}
\]
}

\medskip

Let us make some observations about this result.  First, our Theorem~\ref{thm:main-positive} for the deletion channel is the special case of Theorem~\ref{thm:main-positive-general} obtained by setting $\sigma = \gamma = 0$.  Next, for fixed $\delta$,
\begin{align*}
	\text{if $\delta \leq 1/2$, }& \quad \text{$r$ ranges from $1-\delta$ down to $1/2$ as $\sigma$ ranges from $0$ up to $1$;} \\
    \text{if $\delta \geq 1/2$, }& \quad \text{$r$ ranges from $1-\delta$ \emph{up to} $1/2$ as $\sigma$ ranges from $0$ up to $1$.}
\end{align*}
The second statement is rather peculiar: it implies that when the deletion rate is high, the ability to perform trace reconstruction actually \emph{improves}, the more insertions there are.  Indeed, when we have  deletions only, our ability to do trace reconstruction in time $\exp(O(n^{1/3}))$ is limited to retention probability $\rho \geq \Omega(1)$.  But as soon as the insertion rate $\sigma$ satisfies $\sigma \geq \Omega(1)$, we can do trace reconstruction in time $\exp(O(n^{1/3}))$ as long as the  retention rate $\rho = 1-\delta$ satisfies $\rho \geq \exp(-O(n^{1/3}))$.

\ignore{

\section{Boring technical details} \label{app:boring}

. removing the assumption that the algorithm ``knows'' $\rho$.  (Also handle this issue for the more general channel I guess.)

. the LP recovery algorithm requires one to be able to compute (or approximate) $\mu_\channel(x')$ for a given $x' \in [-1,+1]^n$.  Obviously, one could do this with high probability just by simulating the channel oneself.  But one can also just do deterministically, by calculating.  Well, again, there's the annoyance about the insertion-channel not being bounded.
}
\ignore{

\section{Heavy hitters problem}
Our techniques can be used to resolve yet another basic problem is noisy unsupervised problem, namely the \emph{heavy hitters problem}. In the heavy hitters problem, there is an unknown distribution $\pi$ and the learner has access to samples from $T_{\rho}(\pi)$ where a random sample $y$ from $T_{\rho}(\pi)$ is defined as
\begin{itemize}
\item Sample $x \leftarrow \pi$.
\item Sample $y$ by replacing each bit of $x$ with an unbiased random bit with probability $1-\rho$.
\end{itemize}
In other words, $T_{\rho}(\pi)$ is obtained by applying the Bonami-Beckner operator at noise rate $\rho$ to $\pi$. Given an error parameter $\epsilon$, the task of the algorithm is to output the set $S_{\epsilon}  = \{x : \pi(x) > \epsilon\}$ and for every $x \in S_{\epsilon}$, an estimate $\tilde{\pi}(x)$ such that $|\pi(x) - \tilde{\pi}(x)| \le \epsilon$.
Recently, there has been a lot of work on several variants of this problem. In particular, in work on the \emph{noisy population recovery problem}~\cite{WY12, LZ15, DST16}, it has been shown that assuming that $|\mathsf{supp}(\pi)|=k$, the sample complexity of the above problem is
$k^{O_{\rho}(1)} \cdot \poly(1/\epsilon)$. In fact, under a further mild assumption, this sample complexity can be made algorithmic. However, without the bounded support assumption, there is no non-trivial upper bound on the sample complexity of this problem.

On the other hand, Batman \emph{et~al.}~\cite{BIMP13} showed that in time $\poly(n, (1/\epsilon)^{\log \log (1/\epsilon)})$, it is possible to put a set $\tilde{S}_\epsilon$ of size $\poly (1/\epsilon)^{\log \log (1/\epsilon)}$ such that $S_\epsilon \subseteq \tilde{S}_\epsilon$. Thus, in the current state of the art, it is possible to find (a mild) superset of the \emph{heavy-hitters} i.e. $S_{\epsilon}$. On the other hand, it is not known how to find (even approximately) the mass of the heavy-hitters.
Here, we resolve this question.
\begin{theorem}~\label{thm:hh-upper}
There is an algorithm such that given samples from $T_{\rho} \pi$, it runs in time $2^{n^{1/3}/\rho^{1/3}} \cdot \poly(1/\epsilon)$ and for all $x \in S_{\epsilon}$, outputs $\tilde{\pi}(x)$ such that $|\tilde{\pi}(x) - \pi(x) | \le \epsilon$.
\end{theorem}
\begin{theorem}~\label{thm:hh-lower}
There are two distributions $\pi$ and $\mu$ on $\{0,1\}^n$ such that
$|\pi(0) - \mu(0)| =\Omega(1) $ but $\Vert T_{\rho}(\pi) - T_{\rho} (\mu) \Vert_1 \le 2^{-n^{1/3}/\rho^{1/3}}$.
\end{theorem}

}

\end{document}